\documentclass[12pt,a4paper]{amsart}
\usepackage{amsmath,amsxtra,amssymb,latexsym, amscd}
\usepackage[mathscr]{eucal}
\usepackage{ifpdf}
\usepackage{color}

\theoremstyle{plain}
\newtheorem{theorem}{Theorem}[section]
\newtheorem{lemma}[theorem]{Lemma}
\newtheorem{corollary}[theorem]{Corollary}
\newtheorem{proposition}[theorem]{Proposition}
\theoremstyle{remark}
\newtheorem{remark}{Remark}

\newtheorem{example}[theorem]{Example}

\theoremstyle{definition}
\newtheorem{defn}[theorem]{Definition}

\newcommand{\N}{{\mathbf{N}}}
\newcommand{\Tr}{{\textrm{Tr}}}

\def \R {\mathbb{R}}

\def \C {\mathbb{C}}

\def \N {\mathbb{N}}
\def \cS {\mathcal{S}}
\def \HH {\mathcal{H}}

\begin{document}
	\title{Some applications of Choi polynomials of linear maps}
	\author{Minh Toan Ho}
	\address{Institute of Mathematics, VAST, 18 Hoang Quoc Viet, Hanoi, Vietnam}
	\email{hmtoan@math.ac.vn}

\author{Thanh Hieu Le}
\address{Department of Mathematics and Statistics, Quy Nhon University, 
	170 An Duong Vuong, Quy Nhon Nam Ward, Gia Lai, Vietnam}
\email{lethanhhieu@qnu.edu.vn}

\author{Cong Trinh Le}
\address{Department of Mathematics and Statistics, Quy Nhon University, 
	170 An Duong Vuong, Quy Nhon Nam Ward, Gia Lai, Vietnam}
\email{lecongtrinh@qnu.edu.vn}

\author{Hiroyuki Osaka}
\address{Department of Mathematical Sciences, Ritsumeikan University, Kusatsu, Shiga 525-8577, Japan}
\email{osaka@se.ritsumei.ac.jp}

\date{ \today}

\maketitle

\begin{abstract}
	This paper investigates the properties of Choi polynomials and their fundamental role in the theory of positive linear maps between matrix algebras. By focusing on Hermitian symmetric biquadratic forms, we establish a connection between the positivity of these forms and the structure of positive maps. We specifically explore the construction of indecomposable positive maps in matrix algebras, and their application as entanglement witnesses. Our analysis extends to the detection of Positive Partial Transpose (PPT) entangled states and the classification of edge PPT states in $M_m(\mathbb{C}) \otimes M_n(\mathbb{C})$. Our results provide a refined framework for identifying non-separable states that escape the standard PPT criterion, contributing to the broader understanding of entanglement distillation and quantum information theory.
\end{abstract}


\section{Introduction}

The classification of positive linear maps $\phi: M_m(\mathbb{C}) \to M_n(\mathbb{C})$ remains a central challenge in operator theory and quantum information science. A map is called completely positive if its Choi matrix, defined as $C_{\phi} = \sum_{i,j} e_{ij} \otimes \phi(e_{ij})$, where $\{e_{ij}\}$ is the standard basis for $M_m$, is positive semi-definite. However, the structure of maps that are positive but not completely positive is significantly more intricate. These maps are essential for the detection of quantum entanglement, a task often formulated through the construction of entanglement witnesses (see \cite{Horodecki2009,Keybook}).

A density operator $\rho$ in a bipartite Hilbert space $\mathcal{H}_A \otimes \mathcal{H}_B$ is separable if it can be expressed as a convex combination of product states. If a state is not separable, it is said to be entangled. The Positive Partial Transpose (PPT) criterion, introduced by Peres (\cite{Peres1996}) and the Horodeckis (\cite{Horodecki1997}), serves as a powerful necessary condition for separability: every separable state must remain positive under the partial transposition operator $\Gamma$. While this criterion is sufficient for dimensions $2 \times 2$ and $2 \times 3$, the existence of PPT entangled states (PPTES) in higher dimensions necessitates the use of indecomposable positive maps.

In this work, we focus on a specific class of homogeneous polynomials of degree four, known as Hermitian symmetric biquadratic forms (see \cite{A2011}). As demonstrated in the seminal works of Choi and later refinements by Osaka, these forms are intrinsically linked to the range properties of positive maps. 
This class is a subclass of real-valued homogeneous polynomials of degree four. The properties of Hermitian symmetric biquadratic forms and establishes a rigorous mathematical framework linking polynomial algebra with linear operator theory. The content particularly emphasizes the role of biquadratic forms (hereafter called Choi polynomials) in classifying positive linear maps.

One of the foundational results presented is the existence of a one-to-one correspondence between a linear map $\phi: M_m \to M_n$ and its Choi polynomial, defined as $P_\phi(x,y) = y^*\phi(xx^*)y$. 
Through this correspondence, the properties of the map are fully reflected by the polynomial: The map $\phi$ is positive if and only if its corresponding Choi polynomial $P_\phi(x,y)$ is a positive semidefinite biquadratic form, meaning $P_\phi(x,y) \ge 0$ for all vectors $x \in \mathbb{C}^m$ and $y \in \mathbb{C}^n$.
The decomposability of a map is also completely characterized through the structure of its Choi polynomial. Specifically, a linear map $\phi$ is decomposable if and only if its Choi polynomial $P_\phi(x,y)$ can be expanded as a sum of squares of the moduli of bilinear forms and sesquilinear forms. 
In terms of matrices, this is equivalent to the Gram matrix $W$ of the polynomial $P_\phi$ belonging to the decomposable Gram cone $\mathcal{D}(m,n) = \{W = Q + R^\Gamma : Q \ge 0, R \ge 0\}$, where $\Gamma$ denotes the partial transpose.
The most crucial objective is to provide a systematic method for constructing indecomposable biquadratic forms, which in turn generate indecomposable positive maps. Specifically, we construct the family of polynomials:
\[
p_\epsilon(x,y) = p(x,y) - \epsilon||x \otimes y||^2
\] 
and show that there exists a threshold $\delta > 0$ such that for all $0 < \epsilon \le \delta$, the polynomial $p_\epsilon(x,y)$ remains positive semidefinite but its decomposable structure is broken (it becomes an indecomposable form). This is particularly true when $Q+R^\Gamma$ and $Q^\Gamma+R$ are non-trivial projections.
Due to the 1-1 correspondence, these newly constructed forms $p_\epsilon(x,y)$ are precisely the Choi polynomials of a new class of linear maps. These maps inherit the exact properties of their generating polynomials: they are positive maps but are indecomposable.

By applying the results established in Sections 2 and 3, we re-prove and extend several classes of decomposable and indecomposable maps in Section 4. This section demonstrates the practical application of biquadratic forms and Choi polynomials in characterizing map indecomposability, specifically through the analysis of weighted maps $\Phi_{a;m,n,\epsilon}$ and the construction of maps from edge PPT entangled states.

The paper is organized as follows. Section 2 introduces the notation for biquadratic forms and the Choi-Jamiołkowski isomorphism. Some sufficient and necessary conditions for a biquadratic form to be decomposable/indecomposable. Section 3 discusses the properties of Choi polynomials and their spectral characteristics. In Section 4, we apply these mathematical constructs to the problem of identifying PPTES and classifying edge states in $M_4(\mathbb{C}) \otimes M_4(\mathbb{C})$.

\textbf{Notation:} We denote by $\{e_i\}$ the standard basis for $\C^m$ (and similarly for $\C^n$) and set $e_{ij} = e_i e_j^*$. We define the inner product $\langle \cdot , \cdot \rangle$ on $\C^n$ by $\langle x,y\rangle = y^*x$. Let $M_n$ (resp., $M_n(\R)$) denote the algebra of $n\times n$ matrices with complex (resp., real) entries. Let $\cS_n$ be the space of all symmetric matrices of order $n$ with real coefficients. Let $B(M_m, M_n)$ denote the space of all linear maps from $M_m$ to $M_n$. The Choi matrix of a map $\phi: M_m \to M_n$, denoted by $C_{\phi}$, is defined by $C_{\phi} = \sum_{i,j}e_{ij} \otimes \phi(e_{ij})$. The transpose of $A$ is denoted by $A^t$. The partial transpose is a linear operator on $M_m\otimes M_n$ defined by $(A \otimes B)^{\Gamma} = A \otimes B^t$. Let $PPT[m,n]$ denote the set of all positive matrices in $M_m\otimes M_n$ with a positive partial transpose. Finally, let $\sum^2 K$ denote the set of all finite sums of squares of elements in $K$.
\section{Biquadratic forms}
In this section, we investigate a subclass of real-valued homogeneous polynomials of degree four, namely Hermitian symmetric biquadratic forms. Our aim is to characterize the positivity and decomposability of this class of forms.

\subsection{Real biquadratic forms}
A real biquadratic form is a homogeneous polynomial of the form:
$$
F(x,y) = \sum_{i,j,k,l} c_{ijkl} x_ix_jy_ky_l \quad  (i \le j, k\le l, \ c_{ijkl} \in \R) 
$$
with real indeterminates $x=(x_1,\ldots, x_m), \ y = ( y_1, \ldots, y_n).$ Such an $F$ is said to be positive semidefinite (or rnonnegative) if $F(x,y) \ge 0$ for all $x\in \R^m$,$y\in \R^n.$ 

An old question asks whether, if $F$ is positive semidefinite, there must exist real bilinear forms $f_i$ such that $F = \sum_{i} f_i^2$. Recall a real bilinear form $f$ with real indeterminates $x,y$ is a homogeneous polynomial which can be written as $f(x,y) = y^tAx,$ where $A$ is an $n\times m$ matrix with real coefficients.
The well-known fact that a positive semidefinite real biquadratic form must be a sum of squares of real bilinear forms provided that either $m=2$ or $n=2$ (see \cite[Theorem 1]{Calderon}). However, if $m\ge 3$ and $n\ge 3,$ Man D. Choi \cite{Choi75a} gave an example to show that there exists a positive semidefinite real biquadratic form which is not a sum of square of real bilinear forms.
As application, Choi also used this example to show that there exists a positive linear map from $\cS_3$ to $\cS_3$ which is not completely positive map. Follow this idea, we try to study the relation (and applications) of positive semidefinite biquadratic forms and positive maps.

Let $T: M_m(\R) \longrightarrow M_n(\R)$ be a linear map. Then the Choi polynomial of $T$ is determined by  
$$
P_T (x,y)= y^t T(xx^t)y, \hspace{1cm} \forall x\in \R^m, y\in \R^n.
$$ 
Then $P_T$ is a real biquadratic form. If two linear maps $\phi$ and $\psi$ from $M_m$ to $M_n$ have the same Choi polynomial, then
$y^t\phi(xx^t)y = y^t\psi(xx^t)y$ for all $x\in \R^m, y\in \R^n.$ Hence, $\phi(xx^t) = \psi(xx^t)$ for all $x.$ As a consequence, $\phi(X) =\psi(X)$ for all symmetric 
$X\in M_m(\R).$ However, since $\phi(e_{ij} + e_{ji}) = \psi(e_{ij} + e_{ji}),$ we have $\phi (e_{ij}) + \phi(e_{ji}) =\psi (e_{ij}) + \psi(e_{ji}).$ 


Conversely, given a real biquaratic form $\displaystyle F(x,y) = \sum_{i\le j,k\le l}c_{ijkl} x_ix_jy_ky_l$ 
 in $m$-variable $x$ and $n$-variable $y.$ We can write $F(x,y) = y^tB(x) y,$ where $B(x)$ is a symmetric matrix whose coefficients are quadratic forms in $x.$ Then we can get a linear operator $T$ from $\cS_m$ to $\cS_n$ by $T(xx^t) = B(x).$ This implies that $F(x,y) = y^t T(xx^t) y.$ 
Hence, we get the following remark (we can also see the proof of this remark in \cite{Choi75a}).
\begin{remark}
The correspondence between operators $T: \ \cS_m \to \cS_n $ and real biquadratic forms $P_T $ is one to one. 
\end{remark}

A linear map $T: M_m(\R) \longrightarrow M_n(\R)$ is called a congruence map if there exist $m\times n$ matrices with real coefficients $V_i, i=1,\ldots, k$ such that
$$
T(X) = \sum_{i=1}^k V_i^t X V_i \hspace{1cm} \forall X \in M_m(\R).
$$   
\begin{proposition}[\cite{Choi75a}]\label{real1}
Let $T:\ \cS_m \longrightarrow \cS_n$ be a linear map. 
\begin{itemize}
	\item[(a)] $T$ is positive if and only if $P_T(x,y) \ge 0$ for all $x,y.$
	\item[(b)] $T$ is a congruence map if and only if $P_T$ is a sum of square of real bilinear forms. 
\end{itemize}
\end{proposition} 

Calderon \cite{Calderon} showed that every positive semidefinite real biquadratic form on $\R^m\times \R^n$ must be sum of square of real bilinear forms, provided that $m=2$ or $n=2.$ Combine Proposition \ref{real1} and \cite[Theorem 1]{Calderon}, we get the following corollary.
\begin{corollary}
	A positive linear map $T: \cS_2 \longrightarrow \cS_n$ must be a congruence one.
\end{corollary}
 
\subsection{Biquaratic forms}
\begin{defn}\label{defn1}
	A complex polynomial $p(x,y)$ in $x\in \C^m$ and $y\in \C^n$ is said to be \textit{sesquilinear form} (linear in $y$ and anti-linear in $x$) if it has the form:
	$$
	p(x,y) = \sum_{i,j} p_{ij} \bar{x}_i y_j ,  \hspace{1 cm} p_{ij} \in \C.
	$$ 
	A complex polynomial $p$ is said to be \textit{bilinear form} if it has the form:
	$$
	p(x,y) = \sum_{i,j} p_{ij} x_i y_j  ,\hspace{1 cm} p_{ij} \in \C.
	$$ 
	A complex polynomial $p$ is called \textit{biquadratic form} if it has the form:
	$$
	p(x,y) = \sum_{i,j,k,l} p_{ijkl} x_i\bar{x}_j y_k \bar{y}_l ,  \hspace{1 cm} p_{ijkl} \in \C.
	$$ 
	A biquadratic form $p$ is said to be Hermitian symmetric (in the sense of \cite{A2011}) if $p(x,y) \in \R$ for all $x\in \C^m,$ $y\in \C^n.$  Such a biquadratic form is also called a real-valued form.
\end{defn}
Note that $p(x,y)$ is sesquilinear, meaning it is linear in $y$ and anti-linear in $x$.
A form $p$ is called \textit{positive semidefinite} (psd) if $p(x,y)\ge 0$ for all $x,y.$ Sometimes, we use the term \textit{nonnegative form} to refer to a psd form.

It is straightforward to see that $p$ is a bilinear form if and only if there exists an $m\times n$ matrix $A$ such that
$$
p(x,y) = x^T A y = \langle Ay, \bar{x} \rangle.
$$
Analogously, $p$ is a sesquilinear form if and only if there exists an $m\times n$ matrix $A$ such that
$$
p(x,y) = x^* A y = \langle Ay, x \rangle.
$$
Let $\mathrm{BLF}(m,n)$ (resp., $\mathrm{SLF}(m,n)$) denote the set of all bilinear forms (resp., sesquilinear forms) in $x \in \mathbb{C}^m$ and $y \in \mathbb{C}^n$.

\begin{remark}
	If $p$ is a sesquilinear form and $p(x,y) \in \R$ for every real vectors $x$ and $y,$ then $p(x,y)$ is a real bilinear form.
\end{remark}
\begin{proof}
	We can write 
	$$
	p(x,y) = \sum_{i,j} p_{ij} \bar{x}_i y_j ,   \hspace{1 cm} p_{ij} \in \C.
	$$ 
	For all real vectors $x \in \R^m$ and $y\in \R^n$ we have
	$\displaystyle p(x,y) = \sum_{i,j} p_{ij} x_i y_j \in \R .$ Thus,  $p(x,y) = \overline{p(x,y)}.$ Hence $p_{ij} = \bar{p}_{ij} \in \R.$ This means that $p(x,y)$ is a real bilinear form when $x,y$ are real.
\end{proof}
For $x\in \C^m$ and $y\in \C^n,$
\[
z(x,y):=x\otimes y =\bigl(x_{1}y_{1},\ x_{1}y_{2},\ldots, \ x_{1}y_{n},\ x_{2}y_{1},\ldots,\ x_{m}y_{n}\bigr)^{T}\in \mathbb C^{mn}.
\]
Any finite sum of squares of moduli of \emph{bilinear} forms can be written as
\[
\sum_{r}\bigl|x^{t}A_{r}y\bigr|^{2}
\;=\;
z(x,y)^{*}\,Q\,z(x,y)
\qquad\text{for some } Q\succeq 0,
\]
i.e.,\ $Q$ is the Gram matrix of the bilinear sum of squares.

For sesquilinear forms $x^{*}By$, use the monomial vector
\[
w(x,y):=x \otimes \bar{y}=  \bigl(x_1\overline{y_{1}},\ \ldots,\ x_m \overline{y_{n}} \bigr)^{T}\in\mathbb C^{mn},
\]
and similarly
\[
\sum_{s}\bigl|x^{*}B_{s}y\bigr|^{2}
\;=\;
w(x,y)^{*}\,R\,w(x,y)
\qquad\text{for some } R\succeq 0.
\]
Observe that both $z(x,y)^{*}\,Q\,z(x,y)$ and $w(x,y)^{*}\,R\,w(x,y)$ are nonnegative Hermitian symmetric biquadratic forms.
\begin{lemma}\label{Grammatrix}
	Let $p$ be a biquadratic form given by:
	\[ p(x,y) = \sum_{i,j,k,l} p_{ijkl} x_i\bar{x}_j y_k \bar{y}_l, \qquad  x\in \C^m, y\in \C^n .\]
	There exist unique matrices $Q$ and $R$ such that:
	\[ p(x,y) = z(x,y)^* Q z(x,y) = w(x,y)^* R w(x,y), \]
	where $z(x,y) = x \otimes y$ and $w(x,y) = \bar{x} \otimes y$. 
	The relationship between $Q$ and $R$ is a partial transpose on the first subsystem.
	Furthermore, $p$ is Hermitian symmetric if and only if $Q = Q^*$ (and a similar statement holds for $R$). In this case, the relationship between $Q$ and $R$ is also a partial transpose on the second subsystem.
\end{lemma}
Such a matrix $Q$ in Lemma \ref{Grammatrix} is called a \textit{Gram matrix} of $p.$
\begin{proof}
Let $z(x,y) = x \otimes y$. Its components are $z_{ik} = x_i y_k$, and $(z^*)_{jl} = \bar{x}_j \bar{y}_l$. We can rewrite the form as:
	\[ p(x,y) = \sum_{i,j,k,l} p_{ijkl} (\bar{x}_j \bar{y}_l) (x_i y_k) \]
Define $Q=\sum_{i,j=1}^m\sum_{k,l=1}^n Q^{(i,j)}_{(k,l)} e_{i}e_j^* \otimes f_{k}f_l^*,$ where $Q_{(l,k)}^{(j,i)} = p_{ijkl}.$ Then we obtain $p(x,y) = z^* Q z.$ If $p(x,y) =z(x,y)^* P z(x,y), $ then we get $P_{(j,i)}^{(l,k)} = p_{ijkl} = Q_{(j,i)}^{(l,k)}.$ Similarly, we can show the existence and uniqueness of $R.$ 
	Since $R_{(l,k)}^{(i,j)} = p_{ijkl} = Q_{(l,k)}^{(j,i)}$, we have $R = Q^{T_1}$, where $T_1$ denotes the transpose operation on the first Hilbert space $\mathbb{C}^m$.
	
	
If $p(x,y) \in \mathbb{R}$, then $z^* Q z = (z^* Q z)^* = z^* Q^* z$. Since $z(x,y)$ spans $\mathbb{C}^{mn}$, $Q = Q^*$. The converse follows from the property of Hermitian forms. In this case, $Q^{T_1}= Q^{\Gamma}$ and so $R$ is the partial transpose on the second subsystem of $Q.$
\end{proof}

\subsection{Decomposable biquadratic forms}
\begin{defn}
A biquadratic form $p$ is said to be \textit{decomposable} if it can be written as a finite sum of squares of the moduli of bilinear and sesquilinear forms. A positive semidefinite biquadratic form $p$ is said to be \textit{indecomposable} if it is not decomposable. 

  Let us denote by ${\sum}^2|\mathrm{BLF}| (m,n)$ (resp., ${\sum}^2|\mathrm{SLF}|(m,n)$) the set of all finite sums of the squared moduli of bilinear (resp., sesquilinear) forms in the variables $x\in \C^m$ and $y\in \C^n.$  
  
\end{defn}
Denote
\[
	\mathcal D (m,n) \;=\;\left\{\,W = Q + R^{\Gamma}\ :\ Q\succeq 0,\ R\succeq 0\,\right\}.
\]
Then $\mathcal D (m,n)$ is a cone and called the \emph{decomposable Gram cone}.

\begin{proposition}\label{SOS}
Let $p(x,y) = (x\otimes y)^{*}\,W\,(x\otimes y)$ be a biquadratic form, where the Gram matrix $W=W^*$ is a square matrix of order $mn$ and $x\in \C^m, y\in \C^n.$ Then 
\begin{itemize}
	\item[(i)] The polynomial $p \in {\sum}^2|\mathrm{BLF} (m,n)| $ if and only if $W$ is positive semidefinite.
	\item[(ii)] The polynomial  $p\in {\sum}^2 |\mathrm{SLF}(m,n)|$ if and only if $W^{\Gamma}$ is positive semidefinite.
	\item[(iii)] $p$ is decomposable if and only if its Gram matrix $W$ belongs to the decomposable Gram cone $\mathcal D (m,n)$.
\end{itemize} 
\end{proposition}
\begin{proof}
(i) If $f(x,y)$ is a bilinear form, then we can write
$$
f(x,y) = \sum_{i,j} f_{ij}x_iy_j = \langle x\otimes y, \overline{\operatorname{vec}(f)}\rangle = (x\otimes y)^* {\operatorname{vec}(f)} ,  
$$ 
where $\operatorname{vec}(f)$ the vector of coefficients of $f.$ Hence, if $p \in {\sum}^2|\mathrm{BLF} (m,n)| $, then there are bilinear forms $f_1, \ldots, f_k$ such that 
$$p(x,y) = \sum_{i=1}^k |f_i(x,y)|^2  = \sum_{i=1}^k (x\otimes y)^*{\operatorname{vec}(f_i)}{\operatorname{vec}(f_i)}^*(x\otimes y).  $$
By the uniqueness of the Gram matrix (Lemma \ref{Grammatrix}), $W = {\operatorname{vec}(f_1)}{\operatorname{vec}(f_1)}^* + \cdots + {\operatorname{vec}(f_k)}{\operatorname{vec}(f_k)}^* $ is positive semidefinite. 
Conversely, if $W$ is positive semidefinite, then there is a representation $W = \lambda_1v_1v_1^* + \cdots + \lambda_k v_kv_k^*,$ where each $v_i$ is a unit eigenvector corresponding to the eigenvalue $\lambda_i$ of $W.$ Hence,
$$p(x,y) =  \sum_{i=1}^k \lambda_i (x\otimes y)^{*}\, v_iv_i^* \,(x\otimes y)  = \sum_{i=1}^k \lambda_i |v_i^*(x\otimes y)|^2.$$

The proof of (ii) is the same as that of (i), and (iii) follows from (i), (ii) and Lemma \ref{Grammatrix}. 
\end{proof}
Let $W=W^* \in M_{mn}.$ Define
$$
V_W = \{x\otimes y \ | \ (x\otimes y)^* W (x\otimes y) =0\}.
$$
We say $W_1 \le W_2 $ iff $V_{W_1} \subset V_{W_2}.$ This define a partial order on $\mathcal{D}(m,n).$
\begin{defn}
A decomposable matrix $W\in \mathcal{D}(m,n)$ is said to be an \textit{minimal} if $V_W= \{0\}.$  That is $(x\otimes y)^* W x\otimes y>0$ for every nonzero $x\otimes y,$ $x\in \C^m, y\in \C^n.$
\end{defn}
\begin{lemma}\label{LemmaEdge}
	Let $p(x,y) = (x\otimes y)^*W(x\otimes y)$ be a decomposable form, where $W= Q+R^{\Gamma}$ for some positive semidefinite matrices $Q,R.$ Then the following are equivalent.
\begin{enumerate}
	\item $W$ is minimal.
	\item $\min_{\|x\|=1, \|y\|=1} p(x,y) >0.$
	\item $\ker Q \cap \ker R$ has no nonzero product vectors.
\end{enumerate}
\end{lemma}
\begin{proof}
 The equivalence of (1) and (2) follows from the fact that $p(tx,sy) =t^2s^2p(x,y)$ for any positive numbers $s,t.$ 

Now, we prove the equivalence of (1) and (3). 	We have $$(x\otimes y)^* Q (x\otimes y) = \|Q^{\frac{1}{2}}(x\otimes y)\|^2. $$ 
Using $\mathrm{Tr}(X^{\Gamma}Y)=\mathrm{Tr}(X Y^{\Gamma})$, we also have
\[
 (x\otimes y)^*R(x\otimes y)= \mathrm{Tr}(R^{\Gamma} (x^*x\otimes (y^*y)^T))= \mathrm{Tr}(R^{\Gamma} (x^*x\otimes (\bar{y}^*\bar{y}))) =\|(R)^{\frac{1}{2}}(x\otimes \bar{y})\|^2.
\]
Hence, $$p(x,y)= \|Q^{\frac{1}{2}}(x\otimes y)\|^2 + \|(R)^{\frac{1}{2}}(x\otimes \bar{y})\|^2 >0,$$ for all nonzero $x\otimes y\in \C^m\otimes\C^n$ if and only if $\ker Q \cap \ker R$ has no nonzero product vectors $x\otimes y.$
\end{proof}

Denote $\mathbb{H} (M_m\otimes M_n)$ the space of Hermitian matrices in $(M_m\otimes M_n).$ We equip $\mathbb{H} (M_m\otimes M_n)$ with the Hilbert-Schmidt inner product
\[
\langle X,Y\rangle := \mathrm{Tr}(XY), \enskip
X =X^*,Y=Y^* \in M_m\otimes M_n.
\]
Then partial transpose on the second subsystem $^{\Gamma}$ is self-adjoint, that is, 
for all $X =X^*,Y=Y^*\in M_m\otimes M_n,$ we have
\[
\langle X^{\Gamma},Y\rangle = \langle X, Y^{\Gamma} \rangle.
\]
Recall that for a cone $K \subset \mathbb{H} (M_m\otimes M_n),$
\[
K^* := \{\, X : \mathrm{Tr}(XY)\ge 0, \enskip \forall Y\in K \,\}
\]
denotes its dual cone. Then it is well-known that the cone of positive semidefinite matrices is self-dual:
\[
\{X\succeq 0\}^* = \{X\succeq 0\}.
\]
In addition, the partial transpose on the second subsystem is self-adjoint and involutive, we have the following lemma (see \cite{Keybook} for more detail).
\begin{lemma}\label{dualD}
 $ \mathcal D^* = \{\, X \in M_m\otimes M_n : X\succeq 0 \ \& \  X^{\Gamma}\succeq 0 \}. $
\end{lemma}

Clearly, a real valued polynomial can be written as a difference of sums of squares. On the other hand, given a real valued biquadratic form $f(x,y) = (x\otimes y)^* W x\otimes y $ with Gram matrix $W$, there exist a real number $\delta$ such that $f(x,y) + \delta \|x\otimes y\|^2$ is decomposable. Indeed, we choose $\delta$ to be at least the minimal eigenvalue of $W.$ Then $W+ \delta I$ is positive semidefinite. 
   
In our attempt to construct positive indecomposable biquadratic forms, we usually pay attention to a class of the forms:
$$
F(x,y) - \varepsilon ||x\otimes y || \quad \enskip (x\in \C^m, y\in \C^n),
$$   
where $F(x,y)$ is a decomposable biquadratic form.
\begin{theorem}\label{Theorem1}
	Let $p$ be a decomposable biquadratic form 
	$$
	p(x,y) = (x\otimes y)^* \left( Q  +R^{\Gamma}\right)(x\otimes y)  \quad  \enskip (x\in \C^m, y\in \C^n),
	$$
	where $Q$ and $R$ are positive semidefinite square matrices in $M_m\otimes M_n,$ Let $\varepsilon>0 $ and 
	$p_{\varepsilon} (x,y) := p(x,y) - \varepsilon \| x\otimes y\|^2$.  Then the followings hold.
\begin{itemize}
	\item[(i)] $Q+R^{\Gamma}$ is minimal if and only if there exists a positive number $\delta$ such that  
	the biquadratic form $p_{\varepsilon} $ is positive semidefinite for every $\varepsilon \le \delta.$
	\item[(ii)] $p_{\varepsilon}$ is decomposable if and only if $Q+R^{\Gamma}-\varepsilon I$ is decomposable, where $I$ is the identity matrix of order $mn.$
\end{itemize}
\end{theorem}
\begin{proof}
 Suppose that $W=Q+R^{\Gamma}$ is minimal. By Lemma \ref{LemmaEdge}, we have
	$$\delta:= \inf_{\|x\| =1, \| y\|=1} p(x,y) >0.$$
	For every positive $\varepsilon \le \delta,$ and all nonzero product vector $x\otimes y,$  
	$$
	p_{\varepsilon} (x,y) =\|x\|^2\|y\|^2 \left(p(x_1,y_1) - \varepsilon \right)  \ge \| x\|^2 \|y\|^2 (\delta - \varepsilon) >0, 
	$$
	where $x = \|x\| x_1 $ and $y= \|y\| y_1.$ \\
Conversely, there exists a positive number $\delta$ such that  
the biquadratic form $p_{\varepsilon}$ is positive semidefinite for every $\varepsilon \le \delta.$ Then 
$$
\min_{\| x\| =1, \| y \| =1 } p(x,y) \ge \delta \min_{\| x\| =1, \| y \| =1} \| x\otimes y\|^2 =\delta >0. 
$$
The proof of (ii) follows immediately from Proposition \ref{SOS} and the fact that the Gram matrix of $p_{\varepsilon}$ is $W-\varepsilon I,$ where $W=Q+R^{\Gamma}.$
\end{proof}

\begin{corollary}\label{IndecomBQF1}
Let $p$ and $p_{\varepsilon}$ be given as in Theorem \ref{Theorem1}(i). Suppose further that $\Pi:=Q+R^{\Gamma}$ is a projection and ${\Pi}^{\Gamma}$ a positive contraction. Then $p_{\varepsilon}$ is indecomposable for every $0<\varepsilon <\delta.$ 
\end{corollary}
\begin{proof}
		Let $M= I - \Pi,$ where $I$ is the identity matrix of order $mn.$
	By the hypothesis, $M$ and $M^{\Gamma}$ are positive semidefinite. 
	Moreover,
	\[
	\mathrm{Tr}((\Pi - \varepsilon I) M)
	=\mathrm{Tr}\left((\Pi-\varepsilon I)(I-\Pi)\right)
	=-\varepsilon\,\mathrm{Tr}(M)<0.
	\]
	By Lemma \ref{dualD}, $(\Pi - \varepsilon I)$ cannot lie in the decomposable Gram cone $\mathcal D (m,n).$ By Proposition \ref{SOS}, $p_{\varepsilon}$	is indecomposable for every $0<\varepsilon <\delta.$	
\end{proof}
\begin{example}\label{Exa1}
Let us consider the real symmetric matrix $\Pi\in\mathbb R^{9\times 9}$ satisfying

\[
18\Pi=
\begin{pmatrix}
	11 & -7 & 2 & 2 & 2 & 2 & 2 & 2 & 2\\
	-7 & 11 & 2 & 2 & 2 & 2 & 2 & 2 & 2\\
	2 & 2 & 11 & 2 & 2 & -7 & 2 & 2 & 2\\
	2 & 2 & 2 & 11 & 2 & 2 & -7 & 2 & 2\\
	2 & 2 & 2 & 2 & 2 & 2 & 2 & 2 & 2\\
	2 & 2 & -7 & 2 & 2 & 11 & 2 & 2 & 2\\
	2 & 2 & 2 & -7 & 2 & 2 & 11 & 2 & 2\\
	2 & 2 & 2 & 2 & 2 & 2 & 2 & 11 & -7\\
	2 & 2 & 2 & 2 & 2 & 2 & 2 & -7 & 11
\end{pmatrix}.
\]
Then $\Pi^2=\Pi = \Pi^*$ and $\mathrm{rank}(\Pi)=5$ and $W=\Pi$ satisfies the conditions of Corollary \ref{IndecomBQF1}. Hence, the biquadratic form 
$$F_{\varepsilon} (x,y) = (x\otimes y)^* \Pi (x\otimes y) -\varepsilon \| x\otimes y \|^2\hspace{1cm} (x,y\in \C^3)$$
 is indecomposable for every $0< \varepsilon \le  \min_{\| x\| =1, \| y \| =1 } (x\otimes y)^* \Pi (x\otimes y).$
\end{example}
Next, we will show that $W=\Pi$ satisfies the hypothesis of Corollary \ref{IndecomBQF1}.
\begin{proof}[Proof of Example \ref{Exa1}]
It is straightforward to check that $W=\Pi$ and $W^{\Gamma}$  are projections and rank of $\Pi$ is $5.$
Let $\mathcal N=\ker(\Pi)\subset\mathbb C^9$. A direct computation yields $\dim \mathcal N=4$ and a basis of $\mathcal N$ consisting of the following four vectors (displayed as $3\times 3$ matrices in the $xy^{t}$-reshape):
\[
K_1=
\begin{pmatrix}
	-\frac12&-\frac12&0\\
	0&1&0\\
	0&0&0
\end{pmatrix},\quad
K_2=
\begin{pmatrix}
	-1&-1&1\\
	0&0&1\\
	0&0&0
\end{pmatrix}, \]  
\[ K_3=
\begin{pmatrix}
	-1&-1&0\\
	1&0&0\\
	1&0&0
\end{pmatrix},\quad
K_4=
\begin{pmatrix}
	-1&-1&0\\
	0&0&0\\
	0&1&1
\end{pmatrix}.
\]
In other word,
\[
\ker(\Pi)=\{\mathrm{vec}(K): K\in \mathrm{span}_{\mathbb C}\{K_1,K_2,K_3,K_4\}\}.
\]
Let $z(x,y)=\mathrm{vec}(x y^{t})$ corresponds to the rank-one matrix $xy^{t}$.

\textbf{Claim 1.}
$
\ker(\Pi)\cap \{ \mathrm{vec}(xy^{t}) : x,y\neq 0\}=\{0\}.
$

If $\mathrm{vec}(xy^{t})\in\ker(\Pi)$, then $xy^{t}\in\mathrm{span}\{K_1,K_2,K_3,K_4\}$ and has rank one.
By \textbf{Claim 1}, this forces $xy^{t}=0$, hence $x=0$ or $y=0$.

Now, we prove \textbf{Claim 1.} 
Suppose $K=aK_1+bK_2+cK_3+dK_4$ is a product vector, then $K$ has rank $\le 1.$ We will show that  $a=b=c=d=0$.
	Write 
	\[
	K=
	\begin{pmatrix}
		-\frac a2-b-c-d & -\frac a2-b-c-d & b\\
		c & a & b\\
		c & d & d
	\end{pmatrix}.
	\]
	If $\mathrm{rank}(K)\le 1$, then all $2\times 2$ minors of $K$ vanish. We get
	\[
	c(d-a)=0,\qquad c(d-b)=0,\qquad d(a-b)=0,
	\]
	\[
	b(3a+2b+2c+2d)=0,\qquad b(a+2b+4c+2d)=0,\qquad d(a+4b+2c+2d)=0.
	\]
	We now eliminate the following cases.
	
	\smallskip
	\noindent\emph{Case 1: $b\neq 0$.}
	Then
	\[
	3a+2b+2c+2d=0,\qquad a+2b+4c+2d=0.
	\]
	Subtracting gives $c=a$. Substituting into $a+2b+4c+2d=0$ yields
	\begin{align}
		5a+2b+2d=0. \label{ast} 
	\end{align}
	If $d=0$, then $d(a-b)=0$ is automatic, while $c(d-a)=a(0-a)=-a^2=0$ forces $a=0$ and then (\ref{ast}) gives $b=0$, contradicting $b\neq 0$.
	Hence $d\neq 0$. Then $d(a+4b+2c+2d)=0$ implies
	\[
	a+4b+2c+2d=0.
	\]
	With $c=a$ this becomes $3a+4b+2d=0$. Together with $(\ref{ast})$,
	\[
	\begin{cases}
		5a+2b+2d=0,\\
		3a+4b+2d=0.
	\end{cases}
	\quad
	\textrm{ We have}\
	 a=b.
	\]
	Then $3a+4b+2d=0$ gives $7b+2d=0$, so $d=-\frac{7}{2}b$.
	Now use $c(d-a)=0$ with $c=a=b$:
	\[
	0=c(d-a)=b\Big(-\frac{7}{2}b-b\Big)=-\frac{9}{2}b^2,
	\]
	so $b=0$, contradiction. Therefore $b\neq 0$ is impossible.
	
	\smallskip
	\noindent\emph{Case 2: $b=0$.}
	Then $K$ becomes
	\[
	K=
	\begin{pmatrix}
		-\frac a2-c-d & -\frac a2-c-d & 0\\
		c & a & 0\\
		c & d & d
	\end{pmatrix}.
	\]
	Similar arguments above, we imply that $d=0$.
	
	With $b=d=0$, we have
	\[
	K=
	\begin{pmatrix}
		-\frac a2-c & -\frac a2-c & 0\\
		c & a & 0\\
		c & 0 & 0
	\end{pmatrix}.
	\]
	Repeated the similar arguments above, we imply that $a=b=c=d=0.$


	

\end{proof}

\begin{theorem}[\cite{Lewenstein2000}] \label{edgeW}
	Let $H_A,H_B$ be finite-dimensional complex Hilbert spaces and let
	$\Gamma=\mathrm{id}\otimes T$ denote the partial transpose with respect to a fixed product basis.
	Let $\rho\in\mathbb H(H_A\otimes H_B)$ satisfy the PPT conditions
	\[
	\rho\succeq 0,\qquad \rho^\Gamma\succeq 0.
	\]
	For a product vector $x\otimes y$ we write $(x\otimes y)^\Gamma:=x\otimes \overline y$.
	
	Consider the following statements:
	\begin{enumerate}
		\item[(A)] (\emph{Edge (subtraction) definition}) There do not exist $\varepsilon>0$ and a nonzero product vector
		$x\otimes y$ such that
		\[
		\rho-\varepsilon\, (x\otimes y)(x\otimes y)^* \succeq 0
		\quad\text{and}\quad
		\rho^\Gamma-\varepsilon\,(x\otimes \overline y)(x\otimes \overline y)^* \succeq 0.
		\]
		\item[(B)] (\emph{Range-intersection condition}) There exists no nonzero product vector $x\otimes y$ such that $x\otimes y$ belongs to both the ranges of $\rho$ and $\rho^\Gamma.$
		\item[(C)] (\emph{Kernel-orthogonality condition}) There exists no nonzero product vector $x\otimes y$ such that
		\[
		\langle u,\,x\otimes y\rangle=0\ \ \forall u\in\ker(\rho),
		\qquad
		\langle v,\,x\otimes \overline y\rangle=0\ \ \forall v\in\ker(\rho^\Gamma).
		\]
	\end{enumerate}
	Then
	\[(A)\iff (B)\iff (C).	\]
\end{theorem}
Here, (A) in the theorem is the standard definition of an \emph{edge PPT state} as introduced by
Lewenstein--Kraus--Cirac--Horodecki \cite[Sec.~III]{Lewenstein2000}.

\begin{corollary}\label{EdgePPTEnt}
Let $\rho$ be a PPT entangled edge. Let $W=P+Q^{\Gamma}$, where $P,Q$ be orthogonal projections on $\ker \rho, \ker \rho^{\Gamma},$ respectively.
Then $\delta = \min_{\vert x\vert =1, \vert y \vert =1 } (x\otimes y)^* W (x\otimes y) >0$ and for every $0<\varepsilon\le \delta$, the biquadratic form $F_{\varepsilon}(x,y)= (x\otimes y)^* W (x\otimes y)- \varepsilon \vert x\otimes y \vert^2$ is positive semidefinite  but not decomposable. 
\end{corollary} 
\begin{proof} Set
\[
\delta := \min_{\|x\|=\|y\|=1} \langle x \otimes y, W(x \otimes y) \rangle.
\]
For Horodecki PPT entangled states, by Theorem \ref{edgeW}, this $\delta$ is strictly positive. For any $0 < \varepsilon \le \delta$, the operator $A := W - \varepsilon I_8$ is block-positive (i.e., $P_{\varepsilon} (x,y)= (x\otimes y)^* A (x\otimes y) \ge 0$ for all $x\in \C^2,y\in \C^4$) but detects $\rho$ in the sense that:
\[
\mathrm{Tr}(A\rho) = -8 \varepsilon < 0,
\]
because $P\rho = 0$ and $Q\rho^\Gamma = 0$, hence $\mathrm{Tr}(P\rho) = 0$ and $\mathrm{Tr}(Q^{\Gamma}\rho) = \mathrm{Tr}(Q\rho^\Gamma) = 0.$ Hence, $A \notin \mathcal{D}(m,n).$ Applying Proposition \ref{SOS},  we get the conclusion.
\end{proof}
\section{Choi Polynomials}
Let $\phi: M_m \to M_n$ be a linear map. Let $P_{\phi}$ denote the polynomial defined by:$$P_{\phi}(x,y) = y^*\phi(xx^*)y \hspace{1cm} (x\in \mathbb{C}^m, y\in \mathbb{C}^n).$$We call $P_{\phi}$ \textit{the Choi polynomial of} $\phi$. (Note: M.D. Choi constructed a real biquadratic form on $\mathbb{R}^3 \times \mathbb{R}^3$ that is nonnegative everywhere but is not a sum of squares of real bilinear forms \cite{Choi75a}. This result demonstrated the existence of a positive linear map on $M_3$ that is not decomposable.)
\begin{example}
	Let $F$ be a biquadratic form on $\C^2\times \C^2$ defined by
	$$F(x,y) = 2|x_1|^2|y_1|^2 - 2 \sqrt{-1}\, x_2\bar{x}_1|y_1|^2 + 3\sqrt{-1}\, x_1\bar{x}_2 |y_1|^2 + 3|x_2|^2|y_2|^2. $$
	Then there is a linear map $\phi: \ M_2 \to M_2$ determined by
	$$
	\phi(e_{11}) = \left(\begin{array}{ll}
		2 & 3\sqrt{-1}\\
		-2\sqrt{-1}  & 0
	\end{array}\right),  \quad 
	\phi(e_{22}) = \left(\begin{array}{ll}
		0 & 0\\
		0  & 3
	\end{array}\right) 
	$$
	$$
	\phi(e_{12}) = \left(\begin{array}{ll}
		0& 0\\
		0  & 0
	\end{array}\right), \quad
	\phi(e_{21}) =  \left(\begin{array}{ll}
		0 & 0\\
		0  & 0
	\end{array}\right),
	$$
	such that the Choi polynomial is $F(x,y).$ It is clear that $\phi$ is not self-adjoint and 
	$F\left(\left(\begin{array}{c}1\\
		1
	\end{array}
	\right),
	\left(\begin{array}{c}
		1\\
		0
	\end{array}
	\right) \right)= 2 + \sqrt{-1}$ is not real.
\end{example}

\begin{proposition}\label{11coresp}
	Let $\phi, \psi$ be linear maps from $M_m$ to $M_n.$ Then the following statements hold true.
	\begin{enumerate}
		\item $P_{\lambda \phi + \mu \psi}(x,y) = \lambda P_{\phi}(x,y) + \mu P_{\psi}(x,y)$ for all $x\in \C^m, y \in \C^n,$ and any complex numbers $\lambda, \mu.$
		\item $P_{\phi}(x,y)=0$ for all $x\in \C^m, y \in \C^n$ if and only if $\phi(X) = 0$ for all $X\in M_m.$
		\item The correspondence that associates each linear map $\phi \in B[M_m, M_n]$  with its Choi polynomial $P_{\phi}$ is one-to-one between $B(M_m,M_n)$ and the set of all biquadratic forms in $x\in \C^m, y\in \C^n.$ 
	\end{enumerate}
\end{proposition}
\begin{proof}
	(1) is straightforward.
	
	(2) Suppose $P_{\phi} =0.$ Then $y^*\phi(xx^*)y=0$ for all $x\in \C^m, y \in \C^n.$ Hence, $\phi(xx^*) =0$ for all $x\in \C^m.$ As a consequence, $\phi(X) =0$ for all Hermitian $X\in M_m.$ Any matrix $X\in M_m$ can be decomposed as $X = X_1+ iX_2,$ where the $X_j$ are Hermitian. Hence $\phi (X) =0.$   
	
	(3) Suppose $p$ is a biquadratic forms. Then $p$ can be written 
	$p(x,y) = \sum c_{ijkl} x_i\bar{x}_j y_k \overline{y}_l.$ We define a linear map $\phi:\ M_m\to M_n $ such that the $(l,k)$ entry of $\phi(e_{ij})$ is $c_{ijkl}.$   
	Then, for $x\in \C^m, y\in \C^{n},$ we have
	\begin{align*}
		y^*\phi(xx^*)y & = \sum_{l=1}^n\overline{y_l}f_l^* \phi[(\sum_{i=1}^m x_i e_i)(\sum_{j=1}^m \overline{x}_j e^*_j)] \sum_{k=1}^n y_kf_k \\
		& = \sum_{i,j=1}^m\sum_{k,l=1}^n  x_i\overline{x}_j y_k \overline{y}_l [f_l^* \phi(e_i e_j^*) f_k]\\
		& = \sum_{i,j=1}^m\sum_{k,l=1}^n  c_{ijkl} x_i\overline{x}_j y_k \overline{y}_l = p(x,y).
	\end{align*}
	By (2), the corresponding between $\phi$ and $P_{\phi}$ is one to one and onto.
	
	Indeed, if $\phi_1, \phi_2 \in B(M_m, M_n)$ satisfies $P_{\phi_1} = P_{\phi_2}$, then, $P_{\phi_1-\phi_2} = 0$. Hence, $\phi_1 = \phi_2$.
\end{proof}

The Choi–Jamiołkowski isomorphism establishes a correspondence between a linear map $\phi: M_m \to M_n$ and its Choi matrix $C_{\phi} \in M_m\otimes M_n$, given by $C_\phi = \sum_{i,j=1}^m e_{ij} \otimes \phi(e_{ij})$.
It is straightforward to compute that for all $x \in \C^m$,$y \in \C^n$ 
$$ 
P_{\phi}(x,y) = \langle C_\phi(\overline{x} \otimes y), \overline{x} \otimes y\rangle,
$$
where $\overline{x} = (\overline{x_1}, \dots, \overline{x_m})^t \in \C^m$.

Let us recall (see, e.g., \cite{Keybook, OsakaBook24}) that
$$B_1(M_m \otimes M_n) = \{X = X^* \in M_m \otimes M_n \mid (x \otimes y)^* X (x \otimes y) \ge 0, \quad \forall x \in \mathbb{C}^m, y \in \mathbb{C}^n\}.$$

\begin{theorem}\label{Thrm1}
	Let $\phi$ be a linear map from $M_m $ to $M_n.$ Then the following statements hold.
\begin{enumerate}
	\item $\phi$ is self-adjoint, if and only if $C_{\phi} = C_{\phi}^*$, if and only if $P_{\phi}$ is real-valued for all $x\in \C^m$, $y\in \C^n$ (in other word, $P_{\phi}$ is Hermitian symmetric).
	\item $\phi$ is positive, if and only if $C_{\phi}$ belongs to $B_1(M_m\otimes M_n),$ if and only if $P_{\phi}(x,y) \ge 0$ for all $x,y.$
	\item $\phi$ is completely positive, if and only if $C_{\phi}$ is positive, if and only if the Choi polynomial $P_{\phi}$ of $\phi$ is a sum of squares of sesquilinear forms (of the form $ x^* A y$). 
	\item $\phi$ is completely copositive, if and only if the partial tranpose $C_{\phi}^{\Gamma}$ is positive, if and only if the Choi polynomial $P_{\phi}$ of $\phi$ is a sum of squares of bilinear forms (of the form $x^t B y$).
	\item $\phi$ is decomposable, if and only if $C_{\phi} \in PPT[m,n],$ if and only if the Choi polynomial $P_{\phi}$ of is a sum of squares of bilinear and sesquilinear forms.
\end{enumerate}
\end{theorem}
\begin{proof} The equivalence between these properties of $\phi$ and of the Choi matrix $C_{\phi}$ are well-known (e.g., see \cite[Section 3.3]{OsakaBook24}). Hence, we need to prove the equivalence between $\phi$ and its Choi polynomial. 
	
(1) If $\phi$ is self-adjoint (i.e., it maps a Hermitian matrix to a Hermitian one),then 
	$$ 
\overline{P_{\phi}(x,y)} = (y^*\phi(xx^*)y)^* = y^* (\phi(xx^*))^*y = y^*\phi(xx^*)y \in \R.
	$$
Conversely, if $P_{\phi}(x,y)$ is real for all $x,y$ then $\phi(xx^x)$ is Hermitian for every $x.$ As a consequence, $\phi$ maps a Hermitian matrix to a Hermitian matrix. Hence, it is self-adjoint.
 
(2) This part follows from the definition of Choi polynomial and that $\phi$ is positive iff $\phi(xx^*) \ge 0$ for every $x\in \C^m.$
	
(3) If $\phi (X) = V^* X V$ then 
$$
P_{\phi}(x,y)   =  y^* V^*(xx^*)V y  = |x^*Vy|^2.   
$$
Conversely, if $P_{\phi}(x,y)$ is a sum of squares of sesquilinear forms, for simplicity, assume $P_{\phi}(x,y) = |f(x,y)|^2,$ where $f$ is a sesquilinear form. We can write 
$f(x,y) = x^*Ay,$ where $A$ is an $m\times n$-matrix. Hence,
$$
P_{\phi}(x,y) = (x^*Ay)^*(x^*Ay) = y^*A^*xx^*Ay = y^*\phi(xx^*)y, \quad  \phi (X) = A^*XA.
$$ 
\quad (4) If $\phi (X) = V^* X^t V$ then 
$$
P_{\phi}(x,y)   =  y^* V^*(xx^*)^tV y  = |x^tVy|^2.   
$$
The converse statement is the same argument as the one in (3).

(5) This part follows from (3) and (4). 
\end{proof}

In 1963, St\o rmer \cite{Stormer63} showed that every positive linear map $\phi$ from $M_2$ to $M_2$ is decomposable. By \cite[Theorem 8.2]{Stormer63} every unital extreme linear map
from $M_2$ to $M_2$ is unitarily equivalent to the one of the form:
$$
\left( \begin{array}{ll}
	a & b  \\
	c  & d	  
\end{array}  \right)  \rightarrow 
\left( \begin{array}{ll}
	a & \alpha b + \beta c \\
	\overline{\alpha}c + \overline{\beta}b  & \gamma a+ \epsilon b+ \overline{\epsilon} c+\delta d	  
\end{array} \right),
$$
where $|\epsilon|^2 = 2\gamma (\delta - |\alpha|^2-|\beta|^2)$ in the case $\gamma \ne 0$ and $|\alpha|=1$ or $|\beta| =1$ in the case $\gamma =0.$
We will consider the case where $\gamma = 0$ (so $\epsilon = 0$) and $|\alpha| = 1$. Let us denote this map by $\phi_0$ (including the non-unital case). 
Then, let $X=(x_{ij})$, 
$$
\phi_0(X) = 
\left( \begin{array}{ll}
	x_{11} & \alpha x_{12} + \beta x_{21} \\
	\overline{\alpha}x_{21} + \overline{\beta}x_{12}  & \delta x_{22},	  
\end{array} \right), $$
where $|\alpha| =1$ and $\sqrt{\delta} = {|\beta|+1}.$
\begin{remark}
$\phi_0$ is unital if and only if $\beta = 0$; that is, $\phi_0$ is extreme in the sense of \cite{Stormer63} if and only if $\beta = 0$. In this case, $\phi_0$ is completely positive.
\end{remark}
\begin{proof}
$\phi_0 (1) =1,$ if and only if, $\delta =1,$ if and only if, $\beta =0.$ 
The Choi matrix $C_{\phi_0}$ in this case is
$$C_{\phi_0} = \begin{pmatrix} 1 & 0 & 0 & \alpha \\ 0 & 0 & 0 & 0 \\ 0 & 0 & 0 & 0 \\ \bar{\alpha} & 0 & 0 & 1 \end{pmatrix}.$$
The eigenvalues are $\{0, 2\}.$
\end{proof}
In the case $\beta \ne 0,$ the map $\phi_0$ is not extreme. Using the Choi polynomial, we have the following decomposition of $\phi_0.$
\begin{corollary} 
The positive linear map $\phi_0$ is decomposable but is neither completely positive nor copositive, provided that $\beta \ne 0.$ In addition, $\phi_0$ can be decomposed uniquely as
$$\phi_0 =\phi_1 +\phi_2,$$
where $\phi_1$ is completely positive, $\phi_2$ is completely copositive and  
$$
\phi_1(xx^*) = \left( \begin{array}{ll}
	\delta^{-1/2}|x_1|^2 & x_1\overline{x}_2 \\
	\overline{x}_1x_2   & \sqrt{\delta} |x_{2}|^2,	  
\end{array} \right), \quad  
\phi_2(xx^*) = |\beta| \left( \begin{array}{ll}
\delta^{-1/2}|x_1|^2 & x_2\overline{x}_1 \\
	 \overline{x}_2x_1   & \sqrt{\delta}|x_{2}|^2,	  
\end{array} \right).
$$
\end{corollary} 
\begin{proof}
We have
\begin{align*}
P_{\phi}(x,y) & = |x_1|^2|y_1|^2 + \alpha x_1 \overline{x}_2 \overline{y}_1 y_2 + \overline{\alpha}  \overline{x}_1 x_2y_1\overline{y}_2 + \beta \overline{x}_1x_2 \overline{y}_1 y_2 + \overline{\beta}x_1 \overline{x}_2 y_1\overline{y}_2 + \delta |x_2|^2|y_2|^2.
\end{align*}
Let $A = (a_{ij})$ and $B=(b_{kl})$ be two matrices in $M_2.$ Let
$ F(x,y) := |x^*Ay|^2 + |x^tBy|^2. $
Then 
\begin{align*}
F(x,y) & = \left|{\sum_{i,k}a_{ik}\overline{x}_i y_k}\right|^2 + \left|{\sum_{i,k}b_{ik}{x}_i y_k}\right|^2 \\
       & = \sum_{i,j,k,l} a_{ik}\overline{a}_{jl}\overline{x}_i x_j y_k \overline{y}_l +  \sum_{i,j,k,l} b_{ik}\overline{b}_{jl}{x}_i \overline{x}_j y_k \overline{y}_l.
\end{align*}
 Identify the coefficients of $F(x,y)$ and $P_{\phi_0}$ we get $a_{12} =a_{21} =b_{12}=b_{21} =0$ and 
 \begin{align*}
 	|a_{11}|^2+|b_{11}|^2 & =1,\\
 	|a_{22}|^2+|b_{22}|^2 & = \delta  = (1+|\beta|)^2 ,\\
 	\overline{a}_{11}a_{22} = \alpha, \ & \  \overline{b}_{11} b_{22} =\beta.
 \end{align*}
Therefore, 
$$
A = e^{i\theta_1} \left(\begin{array}{ll}
	\frac{\overline{\alpha}}{\sqrt{1+|\beta|}} & 0 \\
	0 & 0 
  \end{array} \right) \textrm{  and  }\ 
B = e^{i\theta_2} \left(\begin{array}{ll}
	\frac{\sqrt{|\beta|}}{\sqrt{1+|\beta|}} & 0 \\
	0 & \frac{\beta \sqrt{1+|\beta|}}{\sqrt{|\beta|}}
 \end{array} \right) ,
$$
for any $\theta_1, \theta_2.$ However, the corresponding biquadratic forms of $A$ and $B$: 
\begin{align*}
F_1(x,y)=|x^*Ay|^2 & = y^*A^*xx^*Ay = y^*|A|xx^*|A|y, \\
F_2(x,y)=|x^tBy|^2 & = y^*B^*(x^t)^*x^tBy = y^*|B|(xx^*)^t|B|y
\end{align*} 
are independent of $\theta_1$ and $\theta_2.$ Hence, the resulting biquadratic forms $F_1$ and $F_2$ are uniquely determined. As a consequence, by Lemma \ref{11coresp}, we can get the CP ($\phi_1$) and coCP ($\phi_2$) map uniquely. Recovery of the map $\phi_1,\phi_2$ from the biquadratic forms is followed from the identity $\phi_1(xx^*) = A^*xx^*A=  |A| xx^* |A| $ and $\phi_2(xx^*) = |B|(xx^*)^t|B|.$
\end{proof}

Let $\phi:\ M_m \longrightarrow M_n$ be a linear map. Suppose that $\phi (M_m(\R)) \subset M_n(\R).$ Then there exists the restriction of $\phi$ on $M_m(\R)$, denoted by $\phi_{\R}$ which is a linear map from $M_m(\R)$ to $M_n(\R)$ by $\phi_{\R} (X) = \phi (X)$ for $X\in M_m(\R).$
\begin{proposition}\label{Decom_Cong}
Let $\phi:\ M_m \longrightarrow M_n$ be a self-adjoint linear map with $\phi (M_m(\R)) \subset M_n(\R).$  If $\phi$ is decomposable, then the restriction $\phi_{\R}:\ \cS_m \to \cS_n $ is a congruence.
\end{proposition}
\begin{proof}
For an $m\times n$ matrix $V=(v_{ij}),$ we can write $V = Re(V) + iIm(V),$ where $Re(V) = (Re(v_{kl}))$ the real part and $Im(V) = (Im(v_{kl}))$ the imaginary part (which are matrices with real coefficients). Then for any $x\in \R^m$ and $y\in \R^n$, we have
$$ 
y^*V^*xx^*Vy = |x^*Vy|^2 = |x^tRe(V)y|^2+ |x^t Im(V)y|^2 .
$$	
Hence, $y^*V^*xx^*Vy$ is a sum of squares of real bilinear forms. Similarly, $y^*V^*(xx^*)^tVy$ is also a sum of squares of real bilinear forms.
Hence, the Choi polynomial $P_{\phi_{\R}}(x,y)$ is the sum of squares of real bilinear forms and by Proposition \ref{real1}, we get the conclusion.
\end{proof}
\begin{proposition}
Let $\phi:\ M_2 \longrightarrow M_n$ be a linear map with $\phi (M_2(\R)) \subset M_n(\R).$ If $\phi$ is positive then the restriction $\phi_{\R}:\  \cS_2 \to \cS_n$ is a congruence map. In addition, there exists a completely positive map $\phi_1$ from $M_2$ to $M_n$ such that $\phi = \phi_1$ on $\cS_2$ and $(\phi-\phi_1)(e_{12}) =-(\phi-\phi_1)(e_{21})$ is anti-symmetric. 
\end{proposition}
\begin{proof}
Since $\phi$ is positive, its restriction $\phi_{\R}$ on $M_m(\R)$ is positive. By Proposition \ref{real1}, the Choi polynomial of $\phi_{\R}$ is nonnegative, that is $P_{\phi_{\R}}(x,y) \ge 0$ for all $x,y.$ By \cite[Theorem 1]{Calderon}, $P_{\phi_{\R}}$ is a sum of squares of real bilinear forms. By Proposition \ref{real1}, the restriction $\phi_{\R}$ from $\cS_2$ to $\cS_n$ is a congruence map. That is, there are $m\times 2$ real matrices $V_j$ such that
$$\phi_{\R} (Y) = \sum_{j} V_j^t Y V_j \hspace{1cm} \forall Y =Y^t \in M_2(\R).$$
Define $\phi_1$ from $M_2$ to $M_n$ by 
$$\phi_{1} (X) = \sum_{j} V_j^t X V_j \hspace{1cm} \forall X \in M_2.$$
Then $\phi_1 =\phi$ on $\cS_2.$ Let $ A= (\phi -\phi_1) (e_{12}), $ then 
$0=(\phi -\phi_{1}) (e_{12}+e_{21}) = A + (\phi -\phi_1) (e_{21}).$
Hence $-A = (\phi -\phi_1)(e_{12}^*) = A^* = A^t.$
\end{proof}
\begin{remark}
 Let $\phi_{\varepsilon}:\ M_2 \longrightarrow M_4$ be a linear map whose Choi polynomial 
 $$
 P_{\phi_{\varepsilon}}(x,y) = F_{\varepsilon}(x,y) = (x\otimes y)^* \Pi (x\otimes y) -\varepsilon \|x\otimes y\|^2 \qquad (x\in \C^2, y\in \C^4), \quad   0<\varepsilon \le \delta,
 $$
 where $\Pi, \delta$ and $F_{\varepsilon}(x,y)$ are defined in Example \ref{Exa1}. As in Example \ref{Exa1}, the Choi polynomial is positive semidefinite and not decomposable. By Theorem \ref{Thrm1}, $\phi_{\varepsilon}$ is indecomposable. 
 It is clear that $\phi_{\varepsilon} (M_2(\R)) \subset M_4(\R).$
 \end{remark}

\subsection{Application of optimization to checking positivity}
A Choi polynomial of a self-adjoint linear map is a real-valued  biquadratic form. As the same aregument in Lemma \ref{LemmaEdge}, we get the following corollary.
\begin{corollary}
	Let $\phi :\ M_m \longrightarrow M_n $ be a self-adjoint linear map. Then the following statements are equivalent. 
	\begin{itemize}
		\item[(i)]  $\phi$ is positive.
		\item[(ii)]  For every positive real number $r$,
		$$\min_{||x||=||y||=r} P_{\phi}({x},y) \ge 0. $$
		\item[(iii)] There is a real number $r>0$ such that 
		$$\min_{||x||=||y||=r} P_{\phi}({x},y) \ge 0. $$
	\end{itemize}
\end{corollary}
\begin{proof} 	(i)$\Rightarrow$(ii).  By Theorem \ref{Thrm1}(2), $\phi$ is positive if and only if $P_{\phi}(x,y)$ is nonnegative for all $x,y.$ Hence,
	$$\min_{||x||=r, ||y||=r} P_{\phi}({x},y) \ge 0. $$
	(iii)$\Rightarrow$(i)
	Suppose that there is a positive real number $r$ such that $$\min_{||x||=r, ||y||=r} P_{\phi}({x},y) \ge 0. $$ 
	There is a point $(x_0,y_0) \in \C^m\times \C^n$ such that $\| x_0\| =r=\| y_0\|$ and 
	$$\min_{\|x\|=r, \|y\|=r} P_{\phi}(x,y) = P_{\phi}(x_0,y_0)\ge 0. $$
	Since $P_{\phi}(x,y)$ is homogeneous, we have
	$$P_{\phi}(x,y) = \frac{||x||^2 ||y||^2}{r^4} P_{\phi}(\frac{r}{||x||} x, \frac{r}{||y||} y) \ge P_{\phi}(x_0,y_0) \ge 0.$$ 
	The implication (ii) $\Rightarrow$ (iii) is immediate.
\end{proof}
The characterization of positive linear maps are also related to the problem of determining whether a polynomial is bounded below or not. 
\begin{corollary}\label{lowerbound}
	Let $\phi :\ M_m \longrightarrow M_n $ be a self-adjoint linear map. Then $\phi$ is positive if and only if $P_{\phi}$ is bounded below, i.e, 
	$$\inf\{ P_{\phi}(x,y)  \ | \ x\in \C^m, y\in \C^n \} > -\infty. $$
\end{corollary}
Note that if $P_{\phi} $ is Hermitian symmetric and if we replace $x_i= a_i+ b_i\sqrt{-1},$ $y_j=c_j + d_j \sqrt{-1},$ where $a,b \in \R^m$ and $c,d \in \R^n$ then the obtained polynomial $F(a,b,c,d) = P_{\phi}(a+b\sqrt{-1}, c+d \sqrt{-1})$ is a real polynomial in $a,b,c,d$ (with real coefficients). \textit{There are some criteria on lower bounds} of real polynomials, see \cite{Marshall} and the references therein. 
\begin{proof}[Proof of Corollary \ref{lowerbound}]
	If $\phi$ is positive, then $P_{\phi}(x,y) \ge 0$ for all $x, y$ (by Theorem \ref{Thrm1}). Thus, $P_{\phi}$ is bounded below. Conversely, suppose that $P_{\phi}$ is bounded below, but assume for the sake of contradiction that there exists a point $(x_0,y_0) \in \C^m\times \C^n$ such that $P_{\phi}(x_0,y_0) < 0$. Consider the curve defined by $x(t) = tx_0$ and $y(t) = ty_0$ for $t \in \R$. Then, we have
	$$P_{\phi}(tx_0,ty_0) = t^4 P_{\phi}(x_0,y_0) \to -\infty \quad \text{as} \quad t \to \infty.$$
	This contradicts the assumption that $P_{\phi}$ is bounded below.
\end{proof}

\begin{proposition}\label{producteigenvector}
	Let $\phi :\ M_m \longrightarrow M_n $ be a linear map. Suppose that the Choi matrix $C_{\phi}$ has a negative eigenvalue $\lambda<0$ with an eigenvector $\omega.$ 
	\begin{itemize}
		\item[(i)] If the solution set (the algebraic set)
		$$
		S:= \{(x,y) \in \C^m\times \C^n \ | \ [x_1y_1 \cdots  x_1y_n \ x_2y_1 \cdots x_my_1 \cdots x_my_n] = \omega^t  \}
		$$
		is non-empty, then $\phi$ is not positive.
		\item[(ii)] If $\omega$ is a product vector, then $\phi$ is not positive.
	\end{itemize}
\end{proposition}
\begin{proof}
	(i) Suppose $(x_0,y_0) \in S.$ That is, $x_0\otimes y_0= \omega.$ Then
	$$ 
	P_{\phi}(\overline{x}_0,y_0) = (x_0\otimes y_0)^* C_{\phi} x_0\otimes y_0 = \lambda ||x_0\otimes y_0||^2 <0. 
	$$
	Now, by Theorem \ref{Thrm1}, $\phi$ is not positive.\\
	(ii) If $\omega$ is a product vector, then $\omega$ belongs to $S$ mentioned in (i).
\end{proof}

\subsection{Application of positive linear operators to sum of squares}
 
It is clear that a square of an element in $|\textrm{SLF}(m,n)| \cup |\textrm{BLF}(m,n)|$ is a Hermitian biquadratic form. On the other hand, if $\phi$ is self-adjoint, the Choi polynomial $P_{\phi}$ is Hermitian symmetric. Denote by $HBF(m,n)$ the set of all Hermitian symmetric biquadratic forms in $x\in \C^m, y\in \C^n$ and $HBF(m,n)_+$ the subset of all $F(x,y) \in HBF(m,n)$ satisfying $F(x,y) \ge 0$ for all $x,y.$ It is clear that 
\begin{equation}\label{sumsquare}
	{\sum}^2 |\textrm{SLF}(m,n)| + {\sum}^2|\textrm{BLF}(m,n)| \subset HBF(m,n)_+. 
\end{equation}
An interesting problem is to classify the pairs $(m, n)$ for which the reverse inclusion of (\ref{sumsquare}) holds. It is well known that in the real case, this reverse inclusion holds if either $n=2$ or $m=2$, but fails if both $m$ and $n$ are at least 3 (see \cite{Choi75a}). In the complex case, we can apply well-known results from the theory of decomposable maps.
\begin{corollary} The equality
	\begin{equation*}
		{\sum}^2 |\textrm{SLF}(m,n)| + {\sum}^2|\textrm{BLF}(m,n)| = HBF(m,n)_+
	\end{equation*}
	holds if and only if $n+m\le 5.$
\end{corollary}
\begin{proof}
	This follows from Theorem \ref{Thrm1} and the well-known fact that every positive linear map $\phi: \ M_m \to M_n$ is decomposable if and only if $m+n \le 5$ (see \cite{Woronowicz76}).
\end{proof}
\textit{We are interested in characterization of nonnegative biquadratic forms which do (or do not) belong to $\sum^2 |\textrm{SLF}(m,n)| + \sum^2|\textrm{BLF}(m,n)| .$} However, by Theorem \ref{Thrm1}, a real-valued biquadratic form $p$ is decomposable if and only if its corresponding self-adjoint linear map $\phi$ is decomposable, where $P_{\phi} =p.$ 
In Section 2, we study some classes of decomposable/indecomposable biquadratic forms, in particular the case where the Gram matrix of the biquadratic forms can be written as $W-\varepsilon I,$ where $W$ is minimal which is discussed before Lemma \ref{LemmaEdge}. The problem is still open in general .



\section{Examples}
In this section, we present several classes of examples (either new or extending previously well-known ones) of indecomposable linear maps. To achieve this, by Theorem \ref{Thrm1}, we construct indecomposable biquadratic forms $p(x,y)$. Then, the corresponding indecomposable map $\phi$ is determined as follows: the $(k, l)$-entry of $\phi(e_{ij})$ is the coefficient $p_{ijkl}$ of the monomial $x_i\bar{x}_jy_l\bar{y}_k$ in the polynomial $p(x,y)$ (by Proposition \ref{11coresp}). For example, as in   Example \ref{Exa1}, the polynomial  $F_{\varepsilon}$ is indecomposable, so by Theorem \ref{Thrm1}, its corresponding map $\Phi_{\varepsilon} : \ M_3 \to M_3$ is indecomposable for every $0<\varepsilon \le \delta \cong 0.0284.$

\subsection{Indecomposable maps from a given edge PPT entangled state}

We now illustrate the main results of the previous sections by means of a concrete and classical example, namely the $2\otimes 4$ Horodecki family of PPT entangled states. The point of this example is that Theorem \ref{edgeW} and Corollary \ref{EdgePPTEnt} provide a systematic way to construct positive semidefinite but indecomposable biquadratic forms from an edge PPT entangled state, while Theorem \ref{Thrm1} then converts such forms into indecomposable linear maps. Thus the example below serves both as an application of the abstract theory and as an explicit model for the general mechanism. See Horodecki~\cite{Horodecki1997}.
We consider the PPT entangled state $\rho_a  \in M_2(\mathbb C)\otimes M_4(\mathbb C)\cong M_8(\mathbb C)$ in the case of $a=1/2$, as studied in \cite{Horodecki1997}, as follows:
\[
\rho=
\begin{pmatrix}
\frac19 & 0 & 0 & 0 & 0 & \frac19 & 0 & 0 \\
0 & \frac19 & 0 & 0 & 0 & 0 & \frac19 & 0 \\
0 & 0 & \frac19 & 0 & 0 & 0 & 0 & \frac19 \\
0 & 0 & 0 & \frac19 & 0 & 0 & 0 & 0 \\
0 & 0 & 0 & 0 & \frac16 & 0 & 0 & \frac{\sqrt3}{18} \\
\frac19 & 0 & 0 & 0 & 0 & \frac19 & 0 & 0 \\
0 & \frac19 & 0 & 0 & 0 & 0 & \frac19 & 0 \\
0 & 0 & \frac19 & 0 & \frac{\sqrt3}{18} & 0 & 0 & \frac16
\end{pmatrix}.
\]
\begin{proposition}
Let \(P\) and \(Q\) be the orthogonal projections onto \(\ker\rho\) and \(\ker\rho^\Gamma\), respectively, and define
\[
W:=P+Q^\Gamma.
\]
Then the following assertions hold.

\begin{enumerate}
\item[\rm (i)] \(\rho\) is an edge PPT entangled state.
\item[\rm (ii)] The quantity
\[
\delta:=\min_{\|x\|=\|y\|=1}(x\otimes y)^*W(x\otimes y)
\]
is strictly positive.
\item[\rm (iii)] For every \(0<\varepsilon\le \delta\), the biquadratic form
\[
F_\varepsilon(x,y):=(x\otimes y)^*W(x\otimes y)-\varepsilon\|x\otimes y\|^2,
\qquad x\in\mathbb C^2,\ y\in\mathbb C^4,
\]
is positive semidefinite but not decomposable.
\item[\rm (iv)] Consequently, if \(\Phi_\varepsilon:M_2(\mathbb C)\to M_4(\mathbb C)\) is the linear map whose Choi polynomial is
$P_{\Phi_\varepsilon}=F_\varepsilon,$
then \(\Phi_\varepsilon\) is indecomposable.
\end{enumerate}
\end{proposition}

\begin{proof}
It follows from \cite{Horodecki1997} that there exists no product vector $ x \otimes y\in \mathrm{Ran}(\rho_a)$ such that $x \otimes \overline{y} \in \mathrm{Ran}(\rho_a^\Gamma).$  
Hence
Theorem \ref{edgeW} applies, and therefore Corollary \ref{EdgePPTEnt} yields \(\rho\) is an edge PPT entangled state, 
\[
\delta=\min_{\|x\|=\|y\|=1}(x\otimes y)^*W(x\otimes y)>0,
\]
and shows that for every $0<\varepsilon\le\delta$ the form
\[
F_\varepsilon(x,y)=(x\otimes y)^*W(x\otimes y)-\varepsilon\|x\otimes y\|^2
\]
is positive semidefinite but not decomposable.

Now let $\Phi_\varepsilon$ be the unique linear map associated with $F_\varepsilon$
through the Choi-polynomial correspondence. Since
\[
P_{\Phi_\varepsilon}(x,y)=F_\varepsilon(x,y)\ge 0
\enskip  \text{for all }x\in\mathbb C^2,\ y\in\mathbb C^4,
\]
Theorem \ref{Thrm1} (2) implies that $\Phi_\varepsilon$ is positive. Since $F_\varepsilon$
is not decomposable, Theorem \ref{Thrm1} (5) implies that $\Phi_\varepsilon$ is not
decomposable. Therefore $\Phi_\varepsilon$ is an indecomposable positive map.
\end{proof}

\subsection{Indecomposable maps $\Phi_{a;m,n,\varepsilon}$}

Let $m, n \in \mathbb{N}$ with $n \ge m$, and set $r := n - m$. For
$$\varepsilon = (\varepsilon_0, \dots, \varepsilon_r), \qquad 0 < \varepsilon_\alpha \le 1,$$
define
\[
\Phi_{a;m,n,\varepsilon}: M_m(\mathbb{C}) \to M_n(\mathbb{C}), \qquad
\Phi_{a;m,n,\varepsilon}(X) = a \Tr(X)I_n - \sum_{\alpha=0}^{r} \varepsilon_\alpha V_\alpha X V_\alpha^*,
\]
where the isometries $V_\alpha: \mathbb{C}^m \to \mathbb{C}^n$ are given by
\[
V_\alpha e_p = f_{p+\alpha}, \qquad p = 1, \dots, m, \quad \alpha = 0, \dots, r,
\]
and $\{e_p\}$ (resp. $\{f_q\}$) is the standard basis for $\mathbb{C}^m$ (resp. $\mathbb{C}^n$).

If $\varepsilon_{\alpha} =1$ for every $\alpha,$ then the map $\Phi_{a;m,n,\bold{1}}$ is the same as $\Phi_{a;m,n}$ in \cite{Mlynik2025}. As we can see below, using sum of squares arguments and Theorem \ref{Thrm1}, we can reprove some main results in \cite{Mlynik2025} for the weighted maps $\Phi_{a;m,n,\varepsilon}.$ In some cases, we even give necessary and sufficient condition when such a map is decomposable.
 
The Choi polynomial of $\Phi_{a;m,n,\varepsilon}$ is determined as
\[
P_{\Phi_{a;m,n,\varepsilon}}(x,y)
= y^*\Phi_{a;m,n,\varepsilon} (xx^*)y = 
a\|x\|^2\|y\|^2
-
\sum_{\alpha=0}^{r}\varepsilon_\alpha
\left|\sum_{p=1}^m x_p\overline{y_{p+\alpha}}\right|^2, \quad \forall x\in \C^m, y\in \C^n.
\]
\begin{proposition}\label{Phi_eps} \rm
	Let $\Phi_{a;m,n,\varepsilon}:M_m(\mathbb C)\to M_n(\mathbb C) $ as above. Its Choi polynomial is determined by
		\[
	\begin{aligned}
		P_{\Phi_{a;m,n,\varepsilon}}(x,y)
		&=
		\sum_{\alpha=0}^{r}\varepsilon_\alpha
		\sum_{1\le p<q\le m}
		\left|x_p\overline{y_{q+\alpha}}-x_q\overline{y_{p+\alpha}}\right|^2
		+\sum_{p=1}^m\sum_{j=1}^n (a-s_j)|x_p y_j|^2,
	\end{aligned}
	\]
	where
	\[
	s_j:=\sum_{\alpha=\max(0,j-m)}^{\min(r,j-1)}\varepsilon_\alpha,
	\qquad j=1,\dots,n.
	\]
	In particular, if
	\[
	a\ge \max_{1\le j\le n}s_j,
	\]
	then $P_{\Phi_{a;m,n,\varepsilon}}$ is decomposable; hence $\Phi_{a;m,n,\varepsilon}$ is decomposable.
\end{proposition}

\begin{proof}
	Let $x=(x_1,\dots,x_m)\in\mathbb C^m$ and $y=(y_1,\dots,y_n)\in\mathbb C^n$. We have
	\[
	\begin{aligned}
		P_{\Phi_{a;m,n,\varepsilon}}(x,y)
		&=
		y^*\Bigl(a\,\Tr(xx^*)I_n-\sum_{\alpha=0}^{r}\varepsilon_\alpha V_\alpha xx^*V_\alpha^*\Bigr)y\\
		&=
		a\|x\|^2\|y\|^2-\sum_{\alpha=0}^{r}\varepsilon_\alpha\, |y^*V_\alpha x|^2.
	\end{aligned}
	\]
	Now, by the definition of $V_\alpha$,
	we have
	\[
	y^*V_\alpha x=\sum_{p=1}^m x_p\overline{y_{p+\alpha}}.
	\]
	Therefore
	\[
	P_{\Phi_{a;m,n,\varepsilon}}(x,y)
	=
	a\|x\|^2\|y\|^2
	-
	\sum_{\alpha=0}^{r}\varepsilon_\alpha
	\left|\sum_{p=1}^m x_p\overline{y_{p+\alpha}}\right|^2.
	\]
	
	For each $\alpha=0,\dots,r$, define
	\[
	u^{(\alpha)}:=(y_{1+\alpha},\dots,y_{m+\alpha})\in\mathbb C^m.
	\]
	Then
	\[
	\left|\sum_{p=1}^m x_p\overline{y_{p+\alpha}}\right|^2
	=
	|\langle x,u^{(\alpha)}\rangle|^2,
	\]
	so
	\[
	\begin{aligned}
		P_{\Phi_{a;m,n,\varepsilon}}(x,y)
		&=
		\sum_{\alpha=0}^{r}\varepsilon_\alpha
		\Bigl(\|x\|^2\|u^{(\alpha)}\|^2-|\langle x,u^{(\alpha)}\rangle|^2\Bigr)\\
		&\quad+
		\|x\|^2\Bigl(a\|y\|^2-\sum_{\alpha=0}^{r}\varepsilon_\alpha\|u^{(\alpha)}\|^2\Bigr).
	\end{aligned}
	\]
	By the Lagrange identity,
	\[
	\|x\|^2\|u\|^2-|\langle x,u\rangle|^2
	=
	\sum_{1\le p<q\le m}|x_p\overline{u_q}-x_q\overline{u_p}|^2
	\qquad (u\in\mathbb C^m).
	\]
	Applying this to $u=u^{(\alpha)}$ yields
	\[
	\|x\|^2\|u^{(\alpha)}\|^2-|\langle x,u^{(\alpha)}\rangle|^2
	=
	\sum_{1\le p<q\le m}
	\left|x_p\overline{y_{q+\alpha}}-x_q\overline{y_{p+\alpha}}\right|^2.
	\]
	This gives the first sum in the asserted decomposition.
	
	It remains to rewrite the second term. Observe that
	\[
	\|u^{(\alpha)}\|^2=\sum_{p=1}^m |y_{p+\alpha}|^2,
	\]
	so
	\[
	\sum_{\alpha=0}^{r}\varepsilon_\alpha \|u^{(\alpha)}\|^2
	=
	\sum_{\alpha=0}^{r}\varepsilon_\alpha\sum_{p=1}^m |y_{p+\alpha}|^2.
	\]
	Fix $j\in\{1,\dots,n\}$. The term $|y_j|^2$ appears in the inner sum exactly when
	\[
	j=p+\alpha
	\quad\text{for some }p\in\{1,\dots,m\},
	\]
	that is,
	\[
	j-m\le \alpha\le j-1.
	\]
	Since also $0\le \alpha\le r$, the coefficient of $|y_j|^2$ is precisely
	\[
	s_j=\sum_{\alpha=\max(0,j-m)}^{\min(r,j-1)}\varepsilon_\alpha.
	\]
	Hence
	\[
	\sum_{\alpha=0}^{r}\varepsilon_\alpha \|u^{(\alpha)}\|^2
	=
	\sum_{j=1}^n s_j |y_j|^2,
	\]
	and therefore
	\[
	\|x\|^2\Bigl(a\|y\|^2-\sum_{\alpha=0}^{r}\varepsilon_\alpha\|u^{(\alpha)}\|^2\Bigr)
	=
	\sum_{p=1}^m\sum_{j=1}^n (a-s_j)|x_p y_j|^2.
	\]
	Combining the two parts proves the formula
	\[
	\begin{aligned}
		P_{\Phi_{a;m,n,\varepsilon}}(x,y)
		&=
		\sum_{\alpha=0}^{r}\varepsilon_\alpha
		\sum_{1\le p<q\le m}
		\left|x_p\overline{y_{q+\alpha}}-x_q\overline{y_{p+\alpha}}\right|^2
		+\sum_{p=1}^m\sum_{j=1}^n (a-s_j)|x_p y_j|^2.
	\end{aligned}
	\]
	The first sum is a sum of squares of moduli of sesquilinear forms, while the second is a sum of squares of moduli of bilinear forms. Thus, if $a\ge \max_j s_j$, all coefficients $a-s_j$ are nonnegative and the form is decomposable. Consequently, $\Phi_{a;m,n,\varepsilon}$ is decomposable.
\end{proof}

\begin{corollary}[The case $r=1$] \rm
	Let $n=m+1$ and let
	\[
	\Phi_{a;m,m+1,\varepsilon}(X)
	=
	a\,\Tr(X)I_{m+1}
	-\varepsilon_0V_0XV_0^*
	-\varepsilon_1V_1XV_1^*.
	\]
	Then
	\[
	\begin{aligned}
		P_{\Phi_{a;m,m+1,\varepsilon}}(x,y)
		& =
		\varepsilon_0\sum_{1\le p<q\le m}
		|x_p\overline{y_q}-x_q\overline{y_p}|^2
		+ \varepsilon_1\sum_{1\le p<q\le m}
		|x_p\overline{y_{q+1}}-x_q\overline{y_{p+1}}|^2\\
	  	&\quad  + \sum_{p=1}^m (a-\varepsilon_0)|x_p y_1|^2 +
		\sum_{p=1}^m\sum_{j=2}^{m}(a-\varepsilon_0-\varepsilon_1)|x_p y_j|^2\\
		&\quad+
		\sum_{p=1}^m (a-\varepsilon_1)|x_p y_{m+1}|^2.
	\end{aligned}
	\]
	In particular, if $a\ge \varepsilon_0+\varepsilon_1,$
	then $\Phi_{a;m,m+1,\varepsilon}$ is decomposable.
\end{corollary}

\begin{proof}
	This is a specialization of the proposition for $r=1.$
\end{proof}

\begin{corollary}[The case $r=2$] \rm
	Let $n=m+2$ and
	\[
	\Phi_{a;m,m+2,\varepsilon}(X)
	=
	a\,\Tr(X)I_{m+2}
	-\varepsilon_0V_0XV_0^*
	-\varepsilon_1V_1XV_1^*
	-\varepsilon_2V_2XV_2^*.
	\]
	Then 
	\[
	\begin{aligned}
		P_{\Phi_{a;m,m+2,\varepsilon}}(x,y)
		&=
		\sum_{\alpha=0}^{2}\varepsilon_\alpha
		\sum_{1\le p<q\le m}
		\left|x_p\overline{y_{q+\alpha}}-x_q\overline{y_{p+\alpha}}\right|^2
		+\sum_{p=1}^m (a-\varepsilon_0)|x_p y_1|^2\\
		&+
		\sum_{p=1}^m (a-\varepsilon_0-\varepsilon_1)|x_p y_2|^2
		+\sum_{p=1}^m\sum_{j=3}^{m}(a-\varepsilon_0-\varepsilon_1-\varepsilon_2)|x_p y_j|^2\\
		&+
		\sum_{p=1}^m (a-\varepsilon_1-\varepsilon_2)|x_p y_{m+1}|^2
		+\sum_{p=1}^m (a-\varepsilon_2)|x_p y_{m+2}|^2.
	\end{aligned}
	\]
	
Therefore,	if $ a\ge \varepsilon_0+\varepsilon_1+\varepsilon_2,$
	then $\Phi_{a;m,m+2,\varepsilon}$ is decomposable.
	
\end{corollary}

\begin{proof}
	Again this follows from Proposition \ref{Phi_eps} when $r=2.$	
\end{proof}
Given a square matrix $X=X^*,$ let $\lambda_{\max} (X)$ denote the largest eigenvalue of $X.$
\begin{proposition}\label{Phi_eps_Pos} 
	\rm
	 $\Phi_{a;m,n,\varepsilon}$ is positive if and only if
	\[
	a\ge
	\sup_{\|x\|=1}
	\lambda_{\max}\!\left(\sum_{\alpha=0}^{r}\varepsilon_\alpha V_\alpha xx^*V_\alpha^*\right), \qquad \forall x\in \C^m, \|x\|=1.
	\]
\end{proposition}

\begin{proof}
	$\Phi_{a;m,n,\varepsilon}$ is positive if and only if, for every unit vector $x\in \C^m,$  
	\[
	\Phi_{a;m,n,\varepsilon}(xx^*)\ge 0.
	\]
By the definition of $\Phi_{a;m,n,\varepsilon}$, we have
	\[
	\Phi_{a;m,n,\varepsilon}(xx^*)
	=
	aI_n-\sum_{\alpha=0}^{r}\varepsilon_\alpha V_\alpha xx^*V_\alpha^*.
	\]
	Hence, 
		\[ aI_n-\sum_{\alpha=0}^{r}\varepsilon_\alpha V_\alpha xx^*V_\alpha^* \ge 0 \Longleftrightarrow
	a\ge \lambda_{\max}\!\left(\sum_{\alpha=0}^{r}\varepsilon_\alpha V_\alpha xx^*V_\alpha^*\right),
	\qquad\text{for all }\|x\|=1.
	\]
\end{proof}

\begin{corollary}[The case $r=0$] \label{Cor_r0}  \rm
	Suppose $n=m$, so that
	\[
	\Phi_{a;m,m,\varepsilon}(X)=a\,\Tr(X)I_m-\varepsilon_0 X.
	\]
	Then
	\[
	\Phi_{a;m,m,\varepsilon}\text{ is decomposable}
	\iff
	\Phi_{a;m,m,\varepsilon}\text{ is positive}
	\iff
	a\ge \varepsilon_0.
	\]
\end{corollary}

\begin{proof}
	By Proposition \ref{Phi_eps},  $a\ge \varepsilon_0$ implies decomposability.
	Conversely, suppose that $\Phi_{a;m,m,\varepsilon}$ is decomposable, then the map is positive and by Proposition \ref{Phi_eps_Pos}, we have $a\ge \varepsilon_0$. 
\end{proof}

\begin{corollary} \label{Cor_m2_Pos} 
	\rm
	Assume $m=2$ and $n=2+r$. Then $\Phi_{a;2,2+r,\varepsilon}$ is positive if and only if
$	a\ge \lambda_{\max}(J_\varepsilon),$
	where
	\[
	J_\varepsilon=
	\begin{pmatrix}
		\varepsilon_0 & \frac12\sqrt{\varepsilon_0\varepsilon_1} & 0 & \cdots & 0\\
		\frac12\sqrt{\varepsilon_0\varepsilon_1} & \varepsilon_1 &
		\frac12\sqrt{\varepsilon_1\varepsilon_2} & \ddots & \vdots\\
		0 & \frac12\sqrt{\varepsilon_1\varepsilon_2} & \ddots & \ddots & 0\\
		\vdots & \ddots & \ddots & \varepsilon_{r-1} &
		\frac12\sqrt{\varepsilon_{r-1}\varepsilon_r}\\
		0 & \cdots & 0 & \frac12\sqrt{\varepsilon_{r-1}\varepsilon_r} & \varepsilon_r
	\end{pmatrix}.
	\]
\end{corollary}

\begin{proof}
	By the previous proposition, positivity is equivalent to
	\[
	a\ge \sup_{\|x\|=1}
	\lambda_{\max}\!\left(\sum_{\alpha=0}^{r}\varepsilon_\alpha V_\alpha xx^*V_\alpha^*\right).
	\]
	
	Let $x=(x_1,x_2)\in\mathbb C^2$ with $\|x\|=1$. Define
	\[
	v_\alpha:=\sqrt{\varepsilon_\alpha}\,V_\alpha x,\qquad \alpha=0,\dots,r.
	\]
	Then
	\[
	\sum_{\alpha=0}^{r}\varepsilon_\alpha V_\alpha xx^*V_\alpha^*
	=
	\sum_{\alpha=0}^{r} v_\alpha v_\alpha^*.
	\]
	The nonzero eigenvalues of this matrix coincide with the nonzero eigenvalues of the Gram matrix
	\[
	G(x):=\bigl(\langle v_\beta,v_\alpha\rangle\bigr)_{\alpha,\beta=0}^{r}.
	\]
	Since
	\[
	V_\alpha x=x_1f_{\alpha+1}+x_2f_{\alpha+2},
	\]
	the supports of $V_\alpha x$ and $V_\beta x$ are disjoint unless $|\alpha-\beta|\le 1$. More precisely,
	\[
	\langle V_\beta x,V_\alpha x\rangle=
	\begin{cases}
		1,& \alpha=\beta,\\[1mm]
		\overline{x_1}x_2,& \beta=\alpha+1,\\[1mm]
		x_1\overline{x_2},& \alpha=\beta+1,\\[1mm]
		0,& |\alpha-\beta|\ge 2.
	\end{cases}
	\]
	Therefore
	\[
	G(x)=
	\begin{pmatrix}
		\varepsilon_0 & \sqrt{\varepsilon_0\varepsilon_1}\,\overline{x_1}x_2 & 0 & \cdots & 0\\
		\sqrt{\varepsilon_0\varepsilon_1}\,x_1\overline{x_2} & \varepsilon_1 &
		\sqrt{\varepsilon_1\varepsilon_2}\,\overline{x_1}x_2 & \ddots & \vdots\\
		0 & \sqrt{\varepsilon_1\varepsilon_2}\,x_1\overline{x_2} & \ddots & \ddots & 0\\
		\vdots & \ddots & \ddots & \varepsilon_{r-1} &
		\sqrt{\varepsilon_{r-1}\varepsilon_r}\,\overline{x_1}x_2\\
		0 & \cdots & 0 & \sqrt{\varepsilon_{r-1}\varepsilon_r}\,x_1\overline{x_2} & \varepsilon_r
	\end{pmatrix}.
	\]
	Now $ |x_1x_2|\le \frac12 $ 
	because $\|x\|=1$. 
	Thus the maximal possible largest eigenvalue is attained when
$|x_1x_2|=\frac12,$
	that is, when $|x_1|=|x_2|=1/\sqrt2$. For such a choice, $G(x)$ is $J_\varepsilon$:
	\[
	J_\varepsilon=
	\begin{pmatrix}
		\varepsilon_0 & \frac12\sqrt{\varepsilon_0\varepsilon_1} & 0 & \cdots & 0\\
		\frac12\sqrt{\varepsilon_0\varepsilon_1} & \varepsilon_1 &
		\frac12\sqrt{\varepsilon_1\varepsilon_2} & \ddots & \vdots\\
		0 & \frac12\sqrt{\varepsilon_1\varepsilon_2} & \ddots & \ddots & 0\\
		\vdots & \ddots & \ddots & \varepsilon_{r-1} &
		\frac12\sqrt{\varepsilon_{r-1}\varepsilon_r}\\
		0 & \cdots & 0 & \frac12\sqrt{\varepsilon_{r-1}\varepsilon_r} & \varepsilon_r
	\end{pmatrix}.
	\]
	Hence
	\[
	\sup_{\|x\|=1}
	\lambda_{\max}\!\left(\sum_{\alpha=0}^{r}\varepsilon_\alpha V_\alpha xx^*V_\alpha^*\right)
	=
	\lambda_{\max}(J_\varepsilon).
	\]
	Thus positivity, and therefore decomposability, is equivalent to $
	a\ge \lambda_{\max}(J_\varepsilon). $
\end{proof}

\begin{corollary}[The case $m=2$, $r=1$]\rm 
	For
	\[
	\Phi_{a;2,3,\varepsilon}(X)
	=
	a\,\Tr(X)I_3-\varepsilon_0V_0XV_0^*-\varepsilon_1V_1XV_1^*,
	\]
	the following are equivalent:
	\begin{enumerate}
		\item $\Phi_{a;2,3,\varepsilon}$ is positive;
		\item
		$\displaystyle
		a\ge
		\frac{\varepsilon_0+\varepsilon_1+\sqrt{\varepsilon_0^2-\varepsilon_0\varepsilon_1+\varepsilon_1^2}}{2}.
		$
	\end{enumerate}
\end{corollary}

\begin{proof}
	By the preceding proposition, it suffices to compute the largest eigenvalue of
	\[
	J_\varepsilon=
	\begin{pmatrix}
		\varepsilon_0 & \frac12\sqrt{\varepsilon_0\varepsilon_1}\\[1mm]
		\frac12\sqrt{\varepsilon_0\varepsilon_1} & \varepsilon_1
	\end{pmatrix}.
	\]
The largest eigenvalue is
	\[
	\lambda_{\max}(J_\varepsilon)
	=
	\frac{\varepsilon_0+\varepsilon_1+\sqrt{\varepsilon_0^2-\varepsilon_0\varepsilon_1+\varepsilon_1^2}}{2}.
	\]
\end{proof}

\begin{remark}\rm
	The inequality $ a\ge \max_{1\le j\le n}s_j $
	is an explicit sufficient condition for decomposability for arbitrary $m$ and $r$. In general it need not be necessary. However, it becomes exact in some situations, notably when $r=0$.
\end{remark}
\subsubsection{Unweighted decomposability}
In this subsection, we consider the unweighted case $\Phi_{a;m,n,\varepsilon}$ (in the previous subsection) when $\varepsilon_\alpha=1$ for all $\alpha$. In this case,
we write $\Phi_{a;m,n}:= \Phi_{a;m,n, \bold{1}}.$  
Let $m,n\in\mathbb N$ with $n\ge m$, and put $r:=n-m$. Consider
\[
\Phi_{a;m,n}(X)
=
a\,\Tr(X)I_n-\sum_{\alpha=0}^{r}V_\alpha X V_\alpha^*,
\qquad X\in M_m(\mathbb C).
\]
where
\[
V_\alpha e_p=f_{p+\alpha},\qquad p=1,\dots,m,\ \alpha=0,\dots,r.
\]

Then
	\[
	\begin{aligned}
		P_{\Phi_{a;m,n}}(x,y)
		&=
		\sum_{\alpha=0}^{r}\sum_{1\le p<q\le m}
		\left|x_p\overline{y_{q+\alpha}}-x_q\overline{y_{p+\alpha}}\right|^2
		+
		\sum_{p=1}^m\sum_{j=1}^n (a-c_j)|x_p y_j|^2,
	\end{aligned}
	\]
	where
	\[
	c_j:=\#\{\alpha\in\{0,\dots,r\}:1+\alpha\le j\le m+\alpha\},
	\qquad j=1,\dots,n.
	\]
	Equivalently,
	\[
	c_j=
	\min(r,j-1)-\max(0,j-m)+1,
	\qquad j=1,\dots,n.
	\]
	In particular,
	\[
	\max_{1\le j\le n} c_j=r+1,
	\]
	and therefore, if
	\[
	a\ge r+1=n-m+1,
	\]
	then $P_{\Phi_{a;m,n}}$ is a sum of squares of moduli of sesquilinear forms and bilinear forms. Hence $\Phi_{a;m,n}$ is decomposable.
Hence, we get the following corollary.

\begin{corollary}
	\rm 
	Let $m,n\in\mathbb N$ with $n\ge m$, and put $r:=n-m$. Consider
	\[
	\Phi_{a;m,n}:M_m(\mathbb C)\to M_n(\mathbb C),\qquad
	\Phi_{a;m,n}(X)=a\,\Tr(X)I_n-\sum_{\alpha=0}^{r}V_\alpha X V_\alpha^*
	\]
	as above. Then the following hold.
	\begin{enumerate}
		\item[(i)] 
		If $	a\ge r+1=n-m+1,$ 
		then $\Phi_{a;m,n}$ is decomposable.
		\item[(ii)] {The case $r=0$.}
		The following are equivalent:
		\begin{enumerate}
			\item[(a)] $\Phi_{a;m,m}$ is decomposable;
			\item[(b)] $\Phi_{a;m,m}$ is positive;
			\item[(c)] $a\ge 1$.
		\end{enumerate}
		\item[(iii)] {The case $m=2$.}
		The following are equivalent:
		\begin{enumerate}
			\item[(a)] $\Phi_{a;2,2+r}$ is positive;
			\item[(b)] $a\ge 1+\cos\!\Bigl(\dfrac{\pi}{r+2}\Bigr)$.
		\end{enumerate}
	\end{enumerate}
\end{corollary}

\begin{proof}
	Part (i) is exactly the unweighted decomposability criterion (Proposition \ref{Phi_eps})).
	
	(ii) follows from Corollary \ref{Cor_r0}. 
	
	(iii) (m=2)
	By Corollary \ref{Cor_m2_Pos}, we have
	\[
	\Phi_{a;2,2+r}\ \text{is positive}
	\iff
	\lambda_{\max}\! (J)), 
	\]
	where
		\[
	J=
	\begin{pmatrix}
		1 & \frac12 & 0 & \cdots & 0\\
		\frac12 & 1 & \frac12 & \ddots & \vdots\\
		0 & \frac12 & \ddots & \ddots & 0\\
		\vdots & \ddots & \ddots & 1 & \frac12\\
		0 & \cdots & 0 & \frac12 & 1
	\end{pmatrix}
	\in M_{r+1}(\mathbb C).
	\]
	The eigenvalues of this Toeplitz tridiagonal matrix are well known:
	\[
	\lambda_j(J)=1+\cos\!\Bigl(\frac{j\pi}{r+2}\Bigr),
	\qquad j=1,\dots,r+1.
	\]
	Hence
	\[
	\lambda_{\max}(J)=1+\cos\!\Bigl(\frac{\pi}{r+2}\Bigr).
	\]
	Therefore
	\[
	\Phi_{a;2,2+r}\ \text{is positive}
	\iff
	a\ge 1+\cos\!\Bigl(\frac{\pi}{r+2}\Bigr).
	\]
\end{proof}

\subsection{Indecomposable maps based on unextendible product bases}

We next show how the results of the previous sections apply to orthonormal unextendible families of product vectors. This yields a concrete class of positive semidefinite but indecomposable biquadratic forms, and therefore, via the Choi polynomial correspondence, a class of indecomposable positive maps. The argument is an immediate consequence of Corollary \ref{IndecomBQF1} together with the UPB construction in \cite[Theorem~3]{Terhal01}.

\begin{proposition}
Let
\[
E=\{z_1,\dots,z_k\}\subseteq \mathbb C^m\otimes \mathbb C^n
\]
be an orthonormal unextendible family of product vectors, and let
\[
P_E:=\sum_{i=1}^k z_i z_i^*
\]
be the orthogonal projection onto \(H_E:=\operatorname{span}E\). Define
\[
\delta_E:=\min_{\|x\|=\|y\|=1}\langle x\otimes y,\,P_E(x\otimes y)\rangle.
\]
Then \(\delta_E>0\). Moreover, for every \(0<\varepsilon\le \delta_E\), the biquadratic form
\[
P_\varepsilon(x,y):=(x\otimes y)^*(P_E-\varepsilon I)(x\otimes y),
\qquad x\in\mathbb C^m,\ y\in\mathbb C^n,
\]
is positive semidefinite and indecomposable.

Consequently, if \(\Phi_\varepsilon:M_m(\mathbb C)\to M_n(\mathbb C)\) denotes the unique linear map whose Choi polynomial is
$P_{\Phi_\varepsilon}=P_\varepsilon,$
then \(\Phi_\varepsilon\) indecomposable.
\end{proposition}

\begin{proof}
The result follows directly from \cite[Theorem~3]{Terhal01}, Corollary \ref{IndecomBQF1}, and Theorem \ref{Thrm1}.
\end{proof}

\subsection{Indecomposability of the Tanahashi-Tomiyama's map $\tau_{4,1}$}

In the following we reprove the indecomposability of the Tanahashi-Tomiyama's map $\tau_{4,1}$ by sums of squares. The indecomposability of $\tau_{n,k}$ was well-known for $1\leq k \leq n-2$ (see, e.g. \cite{Tomiyama1988, Osaka91, Osaka93, Yamagami93, Ha98}).

Recall that $\tau_{4,1}:M_4(\mathbb C)\to M_4(\mathbb C)$ be the linear map
\[
\tau_{4,1}(X)=3\,\varepsilon(X)+\varepsilon(SXS^*)-X,
\]
where $\varepsilon(X)$ denotes the diagonal part of $X$, and $S=[\delta_{i,j+1}]$ is the cyclic shift matrix.
More explicitly,   
\[
\tau_{4,1}(X)=
\begin{pmatrix}
2x_{11}+x_{44} & -x_{12} & -x_{13} & -x_{14}\\
-x_{21} & 2x_{22}+x_{11} & -x_{23} & -x_{24}\\
-x_{31} & -x_{32} & 2x_{33}+x_{22} & -x_{34}\\
-x_{41} & -x_{42} & -x_{43} & 2x_{44}+x_{33}
\end{pmatrix},
\qquad X=(x_{ij})\in M_4(\mathbb C).
\]

 The Choi polynomial of $\tau_{4,1}$ is determined by
\[
P_{\tau_{4,1}}(x,y)=y^*\tau_{4,1}(xx^*)y,
\qquad x,y\in \mathbb C^4.
\]

Firstly, we show the positivity of $P_{\tau_{4,1}}$.

\begin{proposition} $ P_{\tau_{4,1}}(x,y)\ge 0 $  for all $x,y\in\mathbb C^4$.
\end{proposition}

\begin{proof}
Let
\[
x=(x_1,x_2,x_3,x_4)^T,\qquad y=(y_1,y_2,y_3,y_4)^T,
\]
and set
\[
a_i:=|x_i|^2,\qquad b_i:=|y_i|^2
\qquad (i=1,2,3,4).
\]
Indices are taken cyclically modulo $4$, so $a_0=a_4$.

From the explicit formula for $\tau_{4,1}(xx^*)$, we have
\[
P_{\tau_{4,1}}(x,y)
=
\sum_{i=1}^4 (3a_i+a_{i-1})\,b_i
-\left|\sum_{i=1}^4 \overline{x_i}y_i\right|^2.
\]
Thus it suffices to prove
\[
\left|\sum_{i=1}^4 \overline{x_i}y_i\right|^2
\le
\sum_{i=1}^4 (3a_i+a_{i-1})\,b_i.
\]

Assume first that
\[
3a_i+a_{i-1}>0
\qquad (i=1,2,3,4).
\]
Then
\[
\sum_{i=1}^4 \overline{x_i}y_i
=
\sum_{i=1}^4
\frac{\overline{x_i}}{\sqrt{3a_i+a_{i-1}}}\,
\sqrt{3a_i+a_{i-1}}\,y_i.
\]
By the Cauchy--Schwarz inequality,
\[
\left|\sum_{i=1}^4 \overline{x_i}y_i\right|^2
\le
\left(\sum_{i=1}^4 \frac{a_i}{3a_i+a_{i-1}}\right)
\left(\sum_{i=1}^4 (3a_i+a_{i-1})\,b_i\right).
\]
Therefore it remains to prove that
\[
\sum_{i=1}^4 \frac{a_i}{3a_i+a_{i-1}}\le 1.
\]

Assume now that all $a_i>0$, and define
\(\displaystyle
u_i:=\frac{a_{i-1}}{a_i}>0
\qquad (i=1,2,3,4).
\)
Then
\(
u_1u_2u_3u_4=1,
\)
and
\(
\dfrac{a_i}{3a_i+a_{i-1}}=\dfrac{1}{3+u_i}.
\)
Hence it is enough to show that
\[
\sum_{i=1}^4 \frac{1}{3+u_i}\le 1,
\qquad\text{whenever }u_1u_2u_3u_4=1.
\]
Multiplying both sides by $\prod_{i=1}^4(3+u_i)$, this is equivalent to
\[
\prod_{i=1}^4(3+u_i)-\sum_{i=1}^4 \prod_{j\ne i}(3+u_j)\ge 0.
\]
Since $u_1u_2u_3u_4=1$,  the left-hand side becomes
\[
2\sum_{1\le i<j<k\le 4}u_i u_j u_k
+3\sum_{1\le i<j\le 4}u_i u_j
-26.
\]
Applying the arithmetic--geometric mean inequality, we obtain 
\[
\frac{1}{6}\sum_{1\le i<j\le 4}u_i u_j
\ge
\left(\prod_{1\le i<j\le 4}u_i u_j\right)^{1/6}
=
(u_1u_2u_3u_4)^{1/2}=1,
\]
and 
\[
\frac{1}{4}\sum_{1\le i<j<k\le 4}u_i u_j u_k
\ge
\left(\prod_{1\le i<j<k\le 4}u_i u_j u_k\right)^{1/4}
=
(u_1u_2u_3u_4)^{3/4}=1.
\]
Therefore,
\[
2\sum_{1\le i<j<k\le 4}u_i u_j u_k
+3\sum_{1\le i<j\le 4}u_i u_j
\ge 26,
\]
and hence
\(\displaystyle
\sum_{i=1}^4 \frac{1}{3+u_i}\le 1.
\)
It follows that
\(\displaystyle
\sum_{i=1}^4 \frac{a_i}{3a_i+a_{i-1}}\le 1.
\)
Therefore
\[
P_{\tau_{4,1}}(x,y)
=
\sum_{i=1}^4 (3a_i+a_{i-1})\,b_i
-\left|\sum_{i=1}^4 \overline{x_i}y_i\right|^2
\ge 0.
\]
This proves the claim when all $a_i>0$. If some $a_i=0$, the conclusion follows by continuity. 
\end{proof}

\begin{proposition}
The Choi polynomial
$P_{\tau_{4,1}}(x,y) :=y^*\tau_{4,1}(xx^*)y,$
where $x$ and $y$ are vectors in $\mathbb C^4,$
is indecomposable.
\end{proposition}

\begin{proof}
Let \(x=(x_1,x_2,x_3,x_4)\), \(y=(y_1,y_2,y_3,y_4)\). We now restrict to real variables \(x,y\in\mathbb R^4\), and consider the Choi polynomial (the real case):
\[
\begin{aligned}
p(x,y)
& =  (2x_1^2+x_4^2)y_1^2+(2x_2^2+x_1^2)y_2^2 +(2x_3^2+x_2^2)y_3^2 \\
&\quad +(2x_4^2+x_3^2)y_4^2 -2\sum_{1\le i<j\le 4}x_i x_j y_i y_j.
\end{aligned}
\]
Equivalently,
\begin{equation}
\label{eq:p-real}
\begin{aligned}
p(x,y)
&=2\sum_{i=1}^4 x_i^2y_i^2
-2\sum_{1\le i<j\le 4}x_i x_j y_i y_j
+x_4^2y_1^2+x_1^2y_2^2+x_2^2y_3^2+x_3^2y_4^2.
\end{aligned}
\end{equation}
Assume, for contradiction, that $p$ is a sum of square of real bilinear forms 
\begin{equation}
\label{eq:sos-real}
p(x,y)=\sum_{r=1}^L |F_r(x,y)|^2
\end{equation}
for some real bilinear forms
\[
F_r(x,y)=\sum_{i,j=1}^4 a^{(r)}_{ij}x_i y_j \hspace{1cm} a^{(r)}_{ij} \in \R.
\]
Since each summand in \eqref{eq:sos-real} is nonnegative, every zero of \(p\) is a common zero of all \(F_r\).

We first use eight sparse zeros. From \eqref{eq:p-real}, one checks directly that
\[
p(e_1,e_3)=p(e_1,e_4)=p(e_2,e_1)=p(e_2,e_4)=0,
\]
\[
p(e_3,e_1)=p(e_3,e_2)=p(e_4,e_2)=p(e_4,e_3)=0,
\]
where \(e_1,e_2,e_3,e_4\) are the standard basis vectors of \(\mathbb R^4\). Therefore every \(F_r\) vanishes at these eight points. Writing
\[
F_r(x,y)=\sum_{i,j=1}^4 a^{(r)}_{ij}x_i y_j,
\]
we obtain
\begin{equation}
\label{eq:Fr-first}
\begin{aligned}
F_r(x,y)
&=a^{(r)}_{11}x_1y_1+a^{(r)}_{12}x_1y_2+a^{(r)}_{22}x_2y_2+a^{(r)}_{23}x_2y_3\\
&\quad +a^{(r)}_{33}x_3y_3+a^{(r)}_{34}x_3y_4+a^{(r)}_{41}x_4y_1+a^{(r)}_{44}x_4y_4.
\end{aligned}
\end{equation}
Observe from \eqref{eq:p-real}, that for every sign vector
\[
s=(s_1,s_2,s_3,s_4)\in\{\pm1\}^4
\]
we have $ p(s,s)=0.$
Substituting \(x=y=s\) into \eqref{eq:Fr-first}, we obtain
\[
F_r(s,s)
=
c_r
+a^{(r)}_{12}s_1s_2
+a^{(r)}_{23}s_2s_3
+a^{(r)}_{34}s_3s_4
+a^{(r)}_{41}s_4s_1,
\]
where
\[
c_r:=a^{(r)}_{11}+a^{(r)}_{22}+a^{(r)}_{33}+a^{(r)}_{44}.
\]
Since \(F_r(s,s)=0\) for all \(s\in\{\pm1\}^4\), and the five functions
\[
1,\qquad s_1s_2,\qquad s_2s_3,\qquad s_3s_4,\qquad s_4s_1
\]
are linearly independent on \(\{\pm1\}^4\), it follows that
\[
a^{(r)}_{12}=a^{(r)}_{23}=a^{(r)}_{34}=a^{(r)}_{41}=0,
\qquad
a^{(r)}_{11}+a^{(r)}_{22}+a^{(r)}_{33}+a^{(r)}_{44}=0.
\]
Therefore every \(F_r\) is of the form
\begin{equation}
\label{eq:Fr-diag}
F_r(x,y)=\alpha_r x_1y_1+\beta_r x_2y_2+\gamma_r x_3y_3+\delta_r x_4y_4,
\qquad
\alpha_r+\beta_r+\gamma_r+\delta_r=0.
\end{equation}
Thus each \(F_r\) is a linear combination of the diagonal-difference forms
\[
x_2y_2-x_1y_1,\qquad x_3y_3-x_1y_1,\qquad x_4y_4-x_1y_1.
\]
Now every square \(|F_r(x,y)|^2\) with \(F_r\) of the form \eqref{eq:Fr-diag} expands only into monomials
\[
x_i^2y_i^2,\qquad x_i x_j y_i y_j \quad (i\neq j),
\]
that is, only monomials involving matched pairs \((i,i)\). In particular, no such square can produce any of the shifted monomials
\[
x_4^2y_1^2,\qquad x_1^2y_2^2,\qquad x_2^2y_3^2,\qquad x_3^2y_4^2.
\]
Hence no sum of such squares can contain those monomials.
But the explicit formula \eqref{eq:p-real} shows that \(p\) contains all four of them, each with coefficient \(+1\):
\[
x_4^2y_1^2+x_1^2y_2^2+x_2^2y_3^2+x_3^2y_4^2.
\]
This contradicts \eqref{eq:sos-real}.
Therefore \(p\) is not a sum of squares of real bilinear forms. By Propositions \ref{Decom_Cong} and \ref{real1}, \(P_{\tau_{4,1}}\) is indecomposable.
\end{proof}

\section*{Acknowledgment} Part of this work was completed while the research group participated in the IRCP 2026 program at VIASM, Hanoi, Vietnam.
The first author is partially supported by VAST under grant number CSCL01/27-28. The first and fourth authors are partially supported by JSPS Kakenhi Grant Number JP25K07036.


\begin{thebibliography}{99}
	
	\bibitem{Calderon} A. P. Calderbn, \textit{A note on biquadratic forms,} Linear Algebra Appl. \textbf{7} (1973), 175-177.
	\bibitem{Choi75a} M.D. Choi, \textit{Positive semidefinite biquadratic forms,} Linear Algebra Appl. \textbf{12} (1975) 95-100.
	\bibitem{A2011} J.P. D’Angelo,  \textit{Hermitian analogues of Hilbert’s 17-th problem,} Adv. Math. \textbf{226}(2011), 4607–4637.
	\bibitem{Marshall} M. Ghasemi and M. Marshall, \textit{Lower bounds for a polynomial in terms of its coefficients}, Arch. Math. 95 (2010), 343–353.
	\bibitem{Ha98} K.-C. Ha, Atomic positive linear maps in matrix algebras, \textit{Publ. RIMS.
Kyoto Univ.} 34 (1998), 591-599.
	\bibitem{Horodecki1997} P.~Horodecki, \emph{Separability criterion and inseparable mixed states with positive partial transposition},Phys. Lett. A \textbf{232} (1997), 333-339.
	\bibitem{Horodecki2009} R. Horodecki, P. Horodecki, M. Horodecki and K. Horodecki, \textit{Quantum entanglement}, Rev. Mod. Phys. \textbf{81} (2009), 865-942.
	\bibitem{Keybook} S.H. Kye, Positive Maps in Quantum Information Theory, Springer, 2012.
	\bibitem{Lewenstein2000} M.~Lewenstein, B.~Kraus, J.~I.~Cirac, and P.~Horodecki, \emph{Optimization of entanglement witnesses}, Phys. Rev. A \textbf{62} (2000), 052310. (Preprint: arXiv:quant-ph/0005014.)
	
	\bibitem{Mlynik2025} T. M{\l}ynik, H. Osaka, and M. Marciniak, \emph{Characterization of $k$-positive maps}, Commun. Math. Phys. \textbf{406} (2025), no. 3, Paper No. 62, 11 pp.
	\bibitem{OsakaBook24} M. Moslehian and H. Osaka, Advanced Techniques with Block Matrices of Operators, Birkh{\"a}user, 2024.
	
	\bibitem{Osaka91} H. Osaka, \textit{Indecomposable positive maps in low dimensional matrix algebras}, Linear Algebra Appl. \textbf{153} (1991), 73-83.
	\bibitem{Osaka93} H. Osaka, \textit{A series of absolutely indecomposable positive maps in matrix algebras}, Linear Algebra Appl. \textbf{186} (1993), 45-53.
	\bibitem{Peres1996}  A. Peres, \textit{Separability criterion for density matrices,} Phys. Rev. Lett. \textbf{77} (1996), 1413-1415.
	
	\bibitem{Stormer63} E. St\o rmer, \textit{Positive linear maps of operator algebras,} Acta Math. \textbf{110} (1963), 233-278.
	\bibitem{Tomiyama1988} K. Tanahashi and J. Tomiyama, \textit{Indecomposable maps in matrix algebras,} Can. Math. Bull. \textbf{31}(1988) (3), 308-317.
	\bibitem{Terhal01} B. M. Terhal,\textit{A family of indecomposable linear maps based on entangled quantum states}, Linear Alg. Appl. \textbf{323} (2001), 61-73.
  \bibitem{Yamagami93}S. Yamagami, \textit{Cyclic inequalities}, Proc. Amer. Math. Soc. {\bf 118} (1993), no.2, 521-527
	\bibitem{Woronowicz76} S. L. Woronowicz, \textit{Positive maps of low dimensional matrix algebras,} Rep. Math. Phys. \textbf{10} (1976), 165-183.




		
	
\end{thebibliography}
\end{document}